\documentclass[10pt,final,twocolumn]{IEEEtran}

\usepackage[latin1]{inputenc}
\usepackage[english]{babel}
\usepackage{amsmath,amssymb}
\usepackage{graphicx,graphicx,verbatim}
\usepackage[hypertex]{hyperref}
\usepackage{verbatim}
\usepackage{color}

\newtheorem{thm}{Theorem}
\newtheorem{cor}{Theorem}

\newtheorem{defi}{Definition}
\newtheorem{prop}{Proposition}

\newtheorem{exx}{Example}
\newtheorem{remm}{Remark}

\newenvironment{proposition}{\begin{prop}\rm }{\hfill \hspace*{1pt} \hfill $\lrcorner$ \end{prop}}
\newenvironment{theorem}{\begin{thm}\rm }{\hfill \hspace*{1pt} \hfill $\lrcorner$ \end{thm}}
\newenvironment{definition}{\begin{defi}\rm }{\hfill \hspace*{1pt} \hfill $\lrcorner$ \end{defi}}
\newenvironment{remark}{\begin{remm}\rm }{\hfill \hspace*{1pt} \hfill $\lrcorner$\end{remm}}
\newenvironment{example}{\begin{exx}\rm }{\hfill \hspace*{1pt} \hfill $\lrcorner$ \end{exx}}
\newenvironment{proofof}{\noindent {\em Proof of }}{\hfill \hspace*{1pt}
\hfill $\blacksquare$}
\newenvironment{proof}{\noindent {\em Proof.}}{\hfill \hspace*{1pt} \hfill $\square$}

\newcommand\real{\ensuremath{{\mathbb R}}}
\newcommand\realn{\ensuremath{{\mathbb{R}^n}}}

\newcommand\mymatrix[2]{\left[\begin{array}{#1} #2 \end{array}\right]}
\newcommand{\smallmat}[1]{\left[ \begin{smallmatrix}#1
    \end{smallmatrix} \right]}

\newcommand{\calM}{\mathcal{M}}

\begin{document}

\title{A differential Lyapunov framework \\ for contraction analysis
\thanks{
F. Forni is with the Department of Electrical Engineering and Computer Science, 
University of Li{\`e}ge, 4000 Li{\`e}ge, Belgium, \texttt{fforni@ulg.ac.be}. 
His research is supported by FNRS (Belgian Fund for Scientific Research).
R. Sepulchre is with the University of Cambridge, Department of Engineering, Trumpington Street, Cambridge CB2 1PZ, and with the Department of Electrical Engineering and Computer Science, 
University of Li{\`e}ge, 4000 Li{\`e}ge, Belgium, \texttt{r.sepulchre@eng.cam.ac.uk}.
This paper presents research results of the Belgian Network DYSCO
(Dynamical Systems, Control, and Optimization), funded by the
Interuniversity Attraction Poles Programme, initiated by the Belgian
State, Science Policy Office. The scientific responsibility rests with
its authors. }}
\author{F. Forni, R. Sepulchre}
\date{\today}

\maketitle

\begin{abstract}  
Lyapunov's second theorem is an essential tool for stability analysis of differential equations.
The paper provides an analog theorem  for  incremental stability analysis by lifting the Lyapunov function
to the tangent bundle. The Lyapunov function endows the state-space with a  Finsler structure. Incremental stability is
inferred from infinitesimal contraction of the Finsler metrics through integration along solutions curves.
\end{abstract}

\section{Introduction}

At the core of Lyapunov stability theory  is the realization that a pointwise geometric condition
is sufficient to quantify how solutions of a differential equation approach a specific solution. 
The geometric condition checks that the Lyapunov function, a certain distance from a given point to the target solution,
is doomed to decay along the solution stemming from that point. By integration, the pointwise decay of the 
Lyapunov function forces the asymptotic convergence to the target solution. The basic theorem of Lyapunov has 
led to many developments over the last century, that eventually make the body of textbooks on nonlinear systems 
theory and nonlinear control \cite{Khalil3,Slotine,Isidori,Jancovic96}.
Yet many questions of nonlinear systems theory call for an incremental version of the Lyapunov stability concept, 
in which the convergence 
to a specific target solution is replaced by the convergence or \emph{contraction} between any 
pairs of solutions \cite{Angeli00,Lohmiller1998}. 
Essentially, this stronger property means that solutions forget about their initial condition.
Popular control
applications include tracking and regulation \cite{Pavlov05book,Pavlov2008}, observer design \cite{Rouchon2003,Sanfelice12}, coordination, and synchronization \cite{Wang05},
to cite a few. {Those incremental stability questions are  often reformulated as conventional stability 
questions for a suitable error system, the zero solution of the error system translating the convergence of 
two solutions to each other.} This ad-hoc remedy may be successful in specific situations but it faces unpleasant
obstacles that include both   methodological issues -- such as the issue of transforming a time-invariant 
problem into a time-variant one -- and fundamental issues -- such as the issue of defining a suitable error 
between trajectories --. Those limitations also apply to the Lyapunov characterizations of incremental stability 
that have appeared in the recent years, primarily in the important work of Angeli \cite{Angeli00}.

In a seminal paper \cite{Lohmiller1998}, 
Lohmiller and Slotine advocate a different angle of attack for nonlinear stability analysis. 
Their paper brings the attention of the control community to the basic fact that the distance 
measuring the convergence of two trajectories to each other needs not be constructed explicitly. 
Instead, it can be the integral of an infinitesimal measure of contraction. In other words, the often 
intractable construction of a distance needed for a global analysis can be substituted by a local 
construction. At a fundamental level, this approach brings differential geometry to the rescue of 
Lyapunov theory. 
The contraction concept of Lohmiller and Slotine -- sometimes called ``convergence'' 
in reference to an earlier concept of Demidovich \cite{Pavlov04} -- has been successfully used in a number of 
applications in the recent years \cite{Rouchon2003,Pavlov05book,Russo10,Wang05}. Yet, its connections 
to Lyapunov theory have been scarse,  preventing a vast body of  system theoretic tools to be 
exploited in the framework of contraction theory.

The present paper aims at bridging  Lyapunov theory and contraction theory by formulating a differential 
version of the fundamental second's Lyapunov theorem. Assuming that the state-space is a differentiable 
manifold, the classical concept of Lyapunov function in the (manifold) state-space is lifted to the tangent 
bundle. We call this lifted Lyapunov function a Finsler-Lyapunov function because it endows the differentiable 
manifold with a Finsler structure,
which is precisely what is needed to associate by integration a global 
distance (or Lyapunov function) to the local construction. We formulate a Lyapunov theorem that provides a 
sufficient pointwise geometric condition to quantify incremental stability, that is, how solutions of 
differential equations approach each other. The pointwise properties of the Finsler-Lyapunov function in the 
tangent space guarantees that a suitable (integrated) distance function decays along solutions, proving 
incremental stability.

There are a number of reasons that motivate the Finsler structure as the appropriate { 
differential structure} to study incremental stability. Primarily, it unifies the approach advocated 
by Slotine -- which equips the state-space with a Riemannian structure -- and alternative approaches 
to contraction, such as the recent approach by Russo, Di Bernardo, and Sontag \cite{Russo10,Sontag10yamamoto} based on a matrix measure for the local measure of contraction. 
Examples in the paper further suggest that the Finsler framework will allow to unify 
the application of contraction to physical systems -- typically akin to the Riemannian framework of 
classical mechanics -- and to conic applications -- typically akin to (non-Riemannian) Finsler 
metrics -- such as consensus problems or monotone systems encountered in biology.

A primary motivation to study contraction in a (differential) Lyapunov framework is to make the 
whole body of Lyapunov theory  available to contraction analysis. This is a vast program, only 
illustrated in the present paper by the very first extension of  Lyapunov theorem based on  LaSalle's  
invariance principle. Although we are not aware of a published invariance principle for contraction 
analysis, its formulation  in the proposed differential framework is a straightforward extension of 
its classical formulation and we anticipate this mere extension to be as useful for incremental 
stability analysis as it is for classical Lyapunov stability analysis.

We also include in this paper an extension of the basic theorem to the weaker notion of \emph{horizontal} 
contraction. Horizontal contraction is weaker than contraction in that the pointwise decay of the 
Finsler-Lyapunov function is verified only in a subspace -- called the horizontal subspace -- of 
the tangent space. Disregarding contraction in specific directions is a convenient way to take into 
account symmetry directions along which no contraction is expected. This weaker notion of contraction 
is adapted to many physical systems and to many applications where contraction theory has proven useful, 
such as tracking, observer design, or synchronization. 
Those applications involve one or several copies 
of a given system and only the contraction between the copies and the system trajectories is of interest.

The rest of the paper is organized as follows.
The notation is summarized in Section 
\ref{sec:notation}. Sections \ref{sec:incrementa_stability}, 
\ref{sec:FinsLyap_functs}, \ref{sec:FinsLyap_thm} contain
the core of the differential framework through the introduction 
of the main definitions, results, and related examples.
A detailed comparison with the existing literature is proposed
in Section \ref{sec:literature_comparison}.
Finally, LaSalle's invariance principle and horizontal contraction 
are presented in Sections \ref{sec:LaSalle} and
\ref{sec:horizontal_contraction}, respectively. Conclusions follow.

\section{Notation and preliminaries}
\label{sec:notation}

We present the differential framework on general manifolds by adopting the notation used in 
\cite{AbsMahSep2008} and \cite{DoCarmo1992}.
A \emph{($d$-dimensional) manifold $\mathcal{M}$} is a couple $(\mathcal{M},\mathcal{A}^+)$
where $\mathcal{M}$ is a set and $\mathcal{A}^+$ is a maximal atlas of $\mathcal{M}$
into $\real^d$, such that the topology induced by $\mathcal{A}^+$ is Hausdorff and 
second-countable.
We denote the \emph{tangent space} of $\mathcal{M}$ at $x\in\mathcal{M}$ by $T_x\mathcal{M}$,
and the 
\emph{tangent bundle} of $\mathcal{M}$ by
$T\mathcal{M}=\bigcup_{x\in\mathcal{M}} \{x\}\times T_x\mathcal{M}$.

Given two smooth manifolds $\mathcal{M}_1$ and $\mathcal{M}_2$ of dimension $d_1$ and $d_2$ respectively, 
consider a function $F:\mathcal{M}_1 \to \mathcal{M}_2$ and a point $x\in \mathcal{M}_1$, and 
consider two charts $\varphi_1:\mathcal{U}_x\subset\mathcal{M}_1\to\real^{d_1}$ and $\varphi_2:\mathcal{U}_{F(x)}\subset\mathcal{M}_2\to\real^{d_2}$ 
defined on neighborhoods of $x$ and $F(x)$. We say that $F$ is
of class $C^k$, $k\in\mathbb{N}$,  if the function $\hat{F} = \varphi_2 \circ F \circ \varphi_1^{-1}:\real^{d_1}\to\real^{d_2}$ is of class $C^k$.
We say that $F$ is smooth (i.e. of class $C^\infty$) if $\hat{F}$ is smooth.
The \emph{differential of $F$ at $x$} is denoted by $DF(x)[\cdot]:T_x \mathcal{M}_1 \!\to\! T_{F(x)}\mathcal{M}_2$.
It maps each tangent vector 
$\delta x\in T_x \mathcal{M}_1$ to $DF(x)[\delta x] \in T_{F(x)}\mathcal{M}_2\,$\footnote{
We underline that the syntax $DF(x)[v]$ follows the intuitive meaning of 
\emph{D}ifferential of a \emph{function} $F:\mathcal{M}_1 \to \mathcal{M}_2$, 
\emph{computed} at $x\in \mathcal{M}$ and \emph{applied} to the
tangent vector $v\in T_x\mathcal{M}$. $DF(x)[v]$ is replaced by more compact 
expressions like $dF_xv$ or $F_{*x}v$ in many textbooks. However, we found
that the adopted notation makes the calculations more readable because of the 
clear distinction of the three elements $F$, $x$ and $v$.}. 

{
Given a manifold $\mathcal{M}$ of dimension $d$,
to each chart $\varphi:\mathcal{U}\subset\mathcal{M} \to \real^d$
there corresponds a natural chart for $T\mathcal{M}$ given by
$(\varphi(\cdot), D\varphi(\cdot)[\cdot]) : T\mathcal{U}\subset T\mathcal{M} \to \real^d\times\real^d$.
In particular, for every $x\in\mathcal{M}$, 
let $E_i $ be the $i$-the vector of th canonical basis of $\real^d$, 
then $\{D\varphi^{-1}(\varphi(x))[E_1],\dots,D\varphi^{-1}(\varphi(x))[E_d]\}$
is the natural basis of $T_x \mathcal{M}$.
}

A \emph{curve} $\gamma$ on a given manifold $\mathcal{M}$, is a mapping $\gamma :I \subset \real \to \mathcal{M}$.
A \emph{regular curve} satisfies $D\gamma(s)[1]\neq 0$ for each $s\in I$.
For simplicity we sometime use
$\dot{\gamma}(s)$ or $\frac{d\gamma(s)}{ds}$ to denote $D\gamma(s)[1]$.
Following \cite[Appendix A]{Isidori}, given a $C^1$ and time varying \emph{vector field} $f$ 
{ on the manifold $\mathcal{M}$},
which assigns to each point $x\in \mathcal{M}$ a tangent vector $f(t,x)\in T_x\mathcal{M}$ at time $t$, 
a $C^1$ curve $\gamma:I\to\mathcal{M}$ is an \emph{integral curve} of $f$ if 
$D\gamma(t)[1] = f(t,\gamma(t))$ for each $t\in I$. We say that 
a curve $\gamma:I\to\mathcal{M}$ is a \emph{solution} to the differential equation  
$\dot{\gamma}=f(t,\gamma)$ on $\mathcal{M}$ if $\gamma$ is an integral curve of $f$.

Throughout the paper we adopt the following notation.
$I_n$ denotes the identity matrix of dimension $n$.
Given a vector $v$, $v^T$ denotes the transpose vector 
of $v$. The \emph{span} of a set of vectors $\{v_1,\dots,v_n\}$ is given by
$\mathrm{Span}(\{v_1,\dots,v_n\}) := \{v\,|\,\exists \lambda_1,\dots\lambda_n\in\real \mbox{ s.t. } v = \sum_{i=1}^n \lambda_i v_i\}$.
Given a constant $c\in \real$ we write $\real_{\geq c}$ to denote
the subset of $[c,\infty)\subset\real$. A locally Lipschitz function 
$\alpha:\real_{\geq 0}\rightarrow\real_{\geq 0}$ is said 
to belong to \emph{class} $\mathcal{K}$ if it is strictly increasing and $\alpha(0) = 0$;
it belongs to \emph{class} $\mathcal{K}_{\infty}$ if, moreover, $\lim_{r \rightarrow+\infty}\alpha(r)=+\infty$.
A function 
$\beta:\real_{\geq 0}\times \real_{\geq 0}\to \real_{\geq 0}$ is said
to belong to \emph{class} $\mathcal{KL}$ if 
(i) for each $t \geq 0$, $\beta(\cdot,t)$ is a $\mathcal{K}$ function,
and (ii) for each $s\geq 0$, $\beta(s,\cdot)$ is nonincreasing and 
$\lim_{t\to \infty}\beta(s,t)=0$. 

A \emph{distance} (or \emph{metric}) $d:\mathcal{M}\times\mathcal{M}\to \real_{\geq 0}$ 
on a manifold $\mathcal{M}$ is a positive function that satisfies
$d(x,y) = 0$ if and only if $x=y$, for each $x,y\in \mathcal{M}$ and 
$d(x,z) \leq d(x,y) + d(y,z)$ for each $x,y,z\in \mathcal{M}$.
Throughout the paper we assume that $d$ is continuous with respect to
the manifold topology.
Given a set $\mathcal{S}\subset\mathcal{M}$ 
we say that $\mathcal{S}$ is bounded if $\sup_{x,y\in\mathcal{S}} d(x,y) < \infty$ for any given
distance $d$ on $\mathcal{M}$. The distance between a set $\mathcal{S}$ and a point $x$
is given by ${d}(\mathcal{A},x) := \sup_{y\in \mathcal{A}} {d}(y,x)$.
We say that a curve $\gamma:I\to\mathcal{M}$ is bounded
if its range is bounded.
Given two functions $f:\mathcal{Z}\to\mathcal{Y}$ and $g:\mathcal{X}\to\mathcal{Z}$, 
the \emph{composition} $f\circ g$ assigns to each  
each $p\in\mathcal{X}$ the value $f\!\circ\! g(p) = f(g(p)) \in \mathcal{Y}$.
Given a function $f:\real^n\to\real^m$ where we denote the (matrix of) partial derivatives by 
$\frac{\partial f(x)}{\partial x}$ and we write $\frac{\partial f(x)}{\partial x}_{|x=y}$
for the partial derivatives computed at $y\in\realn$.

\section{Incremental stability and contraction}
\label{sec:incrementa_stability}
{ 
Consider a manifold $\mathcal{M}$} and a differential equation 
\begin{equation}
\label{eq:sys}
 \dot{x} = f(t,x) 
\end{equation}
where $f$ is a $C^1$ vector field { which maps each $(t,x)\in\real\times\mathcal{M}$
to a tangent vector $f(t,x)\in T_x\mathcal{M}$.} 
We denote by $\psi_{t_0}(\cdot,x_0)$ the solution to \eqref{eq:sys}
from the initial condition $x_0\in\mathcal{M}$ at time $t_0$, that is, $\psi_{t_0}(t_0,x_0)=x_0$. 
Throughout the paper, following \cite{Sontag10yamamoto},
we simplify the exposition by considering \emph{forward 
invariant} and connected subsets $\mathcal{C} \subset \mathcal{M}$ 
for \eqref{eq:sys} such that $\psi_{t_0}(\cdot,x_0)$ is \emph{forward complete}
for every $x_0\in \mathcal{C}$, that is,
$\psi_{t_0}(t,x_0) \in \mathcal{C}$ for each $t_0$ and each $t\geq t_0$. 
{For simplicity of the exposition,} we also assume that every two points in $\mathcal{C}$
can be connected by a smooth curve {$\gamma:I\to \mathcal{C}$}.

The following definition characterizes several notions of \emph{incremental stability}:
\begin{definition}
\label{def:incremental_stability}
Consider the differential equation \eqref{eq:sys} on a given manifold $\mathcal{M}$. 
Let $\mathcal{C}\subset \mathcal{M}$ be a forward invariant set
 and $d : \calM \times \calM \to \real_{\geq 0}$ a continuous distance on $\calM$. 
The system \eqref{eq:sys} is 
\begin{description}
\item[\hspace{-0mm}\emph{(IS)}]\emph{incrementally stable} on $\mathcal{C}$ 
(with respect to $d$)
 if there exists a $\mathcal{K}$ function $\alpha$ such that { 
 $\forall x_1,x_2\in\mathcal{C}$  , $\forall t_0 \in \real$, $\forall t\geq t_0$,}
 \begin{equation}
 d(\psi_{t_0}(t,x_1),\psi_{t_0}(t,x_2)) \leq \alpha(d(x_1,x_2))\ ; 
 \end{equation}
 \item[\emph{(IAS)}]\emph{incrementally asymptotically stable} on $\mathcal{C}$
 if it is incrementally stable and 
  { $\forall x_1,x_2\in\mathcal{C}$, $\forall t_0 \in \real$ },
 \begin{equation}
 \label{eq:ias}
 \lim\nolimits\limits_{t\to \infty} d(\psi_{t_0}(t,x_1),\psi_{t_0}(t,x_2)) = 0 \; 
\end{equation}
 \item[\emph{(IES)}]\emph{incrementally exponentially stable} on $\mathcal{C}$ 
 if there exist a distance $d$, $K\geq 1$, and $\lambda>0$ such that 
 { $\forall x_1,x_2\in\mathcal{C}$, $\forall t_0 \in \real$, $\forall t\geq t_0$,}
\begin{equation} 
\label{eq:exp_stability}
 d(\psi_{t_0}(t,x_1),\psi_{t_0}(t,x_2)) \leq K e^{-\lambda(t-t_0)} d(x_1,x_2). 
 \end{equation}
 \vspace{-6mm}
\end{description}
\end{definition}
These definitions are incremental versions of classical notions of 
stability, asymptotic stability and exponential stability \cite[Definition 4.4]{Khalil3}, and they 
reduce to those notions the metric space $(\calM,d)$ is complete and 
when either $x_1$ or $x_2$ is an equilibrium of \eqref{eq:sys}.
\emph{Global}, \emph{regional}, and \emph{local} notions 
of stability are specified through the definition of the set $\mathcal{C}$. For example, we say
that \eqref{eq:sys} is incrementally globally asymptotically stable when $\mathcal{C}=\mathcal{M}$.
Note that both (IS) and (IES) properties are uniform with respect to $t_0$.

For $\mathcal{M}=\realn$ and for distances given by norms on $\realn$, the notions
of incremental stability and incremental asymptotic stability given above are equivalent to
the notions of incremental stability and attractive incremental stability of 
\cite[Definition 6.22]{Leine08}, respectively. For $\mathcal{C} = \realn$, the notion of
incremental asymptotic stability is weaker than the notion of incremental global asymptotic stability 
of \cite[Definition 2.1]{Angeli00}, since the latter requires uniform attractivity.

Incremental stability of a dynamical system has been previously characterized by
a suitable extension of Lyapunov theory \cite{Angeli00}. For $\mathcal{M}=\realn$, 
the existence of a Lyapunov function decreasing along any pair of solutions 
is a sufficient condition for incremental stability \cite[Theorem 6.30]{Leine08}. 
The key fact is in recognizing the equivalence between 
the incremental stability of $\dot{x} =  f(t,x)$, $x\in \realn$,
and the stability of the set $\mathcal{A} := \{(x_1,x_2)\in \real^{2n}\,|\, x_1=x_2\}$ 
for the extended system $\dot{x}_1 =  f(t,x_1)$,
$\dot{x}_2 =  f(t,x_2)$. As a direct consequence, incremental asymptotic stability
is inferred from the existence of a Lyapunov function $V(x_1,x_2$) for the set $\mathcal{A}$
with (uniformly) negative derivative along the vector field $f(t,x_1)$, $f(t,x_2)$, for any pair $x_1,x_2$. 
The extension to general manifolds is immediate.

\section{Finsler-Lyapunov functions}
\label{sec:FinsLyap_functs}

This section introduces a concept of Lyapunov function in the tangent bundle $T\mathcal{M}$ of a manifold $\mathcal{M}$.
{
\begin{definition}
\label{def:LyapFins_function}
 Consider a manifold $\mathcal{M}$. 
 A $C^1$ function $V:T\mathcal{M} \to \real_{\geq 0}$ that maps every 
 $(x,\delta x) \in T\mathcal{M}$ to $V(x,\delta x)\in \real_{\geq 0}$, 
 is a \emph{candidate Finsler-Lyapunov function for} \eqref{eq:sys} if
 there exist $c_1,c_2\in \real_{\geq 0}$, $p\in\real_{\geq 1}$, 
 and (a Finsler structure) $F:T\mathcal{M} \to \real_{\geq 0}$ such that,
 $\forall (x,\delta x) \in T\mathcal{M}$,
 \begin{equation}
 \label{eq:FinsLyap_bounds}
  c_1 F(x,\delta x)^p  \leq  V(x,\delta x) \leq  c_2 F(x,\delta x)^p.
 \end{equation}
$F$ satisfies the  following conditions:
 \begin{itemize}
  \item[(i)] $F$ is a $C^1$ function for each $(x,\delta x)\in T\mathcal{M}$ such that $\delta x \neq 0$;
  \item[(ii)] $F(x,\delta x) >0$ for each $(x,\delta x)\in T\mathcal{M}$ such that $\delta x \neq 0$;
  \item[(iii)] $F(x,\lambda \delta x) = \lambda F(x,\delta x)$ for each $\lambda \geq 0$
  and each $(x,\delta x)\in T\mathcal{M}$ (homogeneity); 
  \item[(iv)] $F(x,\delta x_1+\delta x_2) < F(x,\delta x_1) + F(x,\delta x_2)$ for each 
  $(x,\delta x_1),(x,\delta x_2) \in T\mathcal{M}$ such that 
  $\delta x_1\neq \lambda \delta x_2$ for any given  $\lambda \in \real$ (strict convexity).
 \end{itemize}  \vspace{-5mm}
\end{definition}
}
For each $x\in\mathcal{M}$, $V$ is a
measure of the length of the tangent vector $\delta x\in T_x\mathcal{M}$. 
The reason to call such a function 
$V$ a ``Finsler-Lyapunov function'' is that it 
combines the properties of a Lyapunov function and of a Finsler structure. 
The connection with classical Lyapunov functions is at methodological level: 
a candidate Finsler-Lyapunov function $V$ is an abstraction on the system tangent bundle $T \mathcal{M}$,
used to characterize the asymptotic behavior of the system trajectories by looking directly at the vector field $f(t,x)$. { Indeed, $V$ will be used as a  Lyapunov function for the
variational system associated to \eqref{eq:sys}.}
\eqref{eq:FinsLyap_bounds}, combined to the fact that 
$F(x,\cdot)$ defines an asymmetric norm $|\cdot|_x := F(x,\cdot)$ 
in each tangent space $T_x\mathcal{M}$,
 emphasizes the analogies between Finsler-Lyapunov functions and 
classical Lyapunov functions.
Note that the continuous differentiability of $V$ can be relaxed 
as in classical Lyapunov theory, see Remark \ref{rem:piecewise_diff} below. 
 In a similar way, the restriction to time-invariant functions $V$ is only for notational 
convenience but all the results
of the paper extend in a straightforward manner to
time-varying functions $V$ \footnote{{Except Section \ref{sec:LaSalle}, where
the extension requires time periodicity, as in classical LaSalle relaxations of Lyapunov Theory.}}. 

The connection with Finsler structures is 
provided by Items (i)-(iv), which make
$F$ a Finsler structure on $\mathcal{M}$ \cite{Tamassy08}.
Positiveness, homogeneity, and strict convexity of $F$ guarantee 
that $F(x,\cdot)$ is a (possibly asymmetric) Minkowski norm in each tangent space.
Thus, the length of any curve $\gamma$ {induced} by $F$ is independent on 
orientation-preserving  reparameterizations of $\gamma$.

The relation \eqref{eq:FinsLyap_bounds} 
between a candidate Finsler-Lyapunov function $V$
and the associated Finsler structure $F$ 
is a key property for the deduction of incremental stability.
This is because {$F$ induces} a well-defined distance on $\mathcal{M}$
via integration. Following \cite[p.145]{Shen00},
\begin{definition}
\label{def:distance}[\emph{Finsler distance}]
Consider a candidate Finsler-Lyapunov function $V$ on the manifold $\mathcal{M}$
{ and the associated Finsler structure $F$ in Definition \ref{def:LyapFins_function}}.
For any subset $\mathcal{C} \subset \mathcal{M}$ and any two points 
$x_1,x_2\in \mathcal{M}$, let $\Gamma(x_1,x_2)$ be the collection of
piecewise $C^1$ curves $\gamma:I \to \mathcal{C}$,
$I:=\{s \in \real \,|\,0 \leq s \leq 1\}$,  $\gamma(0) = x_1$, and 
$\gamma(1) = x_2$. 

The \emph{distance} (or metric) $d:\mathcal{M}\times \mathcal{M} \to \real_{\geq 0}$ 
{induced by $F$} satisfies
\begin{equation}
 \label{eq:induced_distance_main}
 d(x_1,x_2) := \inf_{\Gamma(x_1,x_2)} \int_I F(\gamma(s),\dot{\gamma}(s)) ds.
 \vspace{-5mm}
\end{equation} \end{definition}
We consider curves whose domain is restricted to $0 \leq s \leq 1$ because
any distance {induced by $F$} is independent from any orientation-preserving 
reparameterization of curves. 
With a slight abuse of notation, in \eqref{eq:induced_distance_main} 
we write $\dot{\gamma}(s)=D\gamma(s)[1]$ to denote the directional derivative 
of a given piecewise $C^1$ function 
$\gamma$ at $s$, implicitly assuming that the differential 
is computed only where the function is differentiable. 
Points of non-differentiability 
characterize a set of measure zero, which can be neglected at integration.

\begin{example}
We review specific classes of candidate Finsler-Lyapunov
functions and classical distance functions.
Consider $\mathcal{C}= \mathcal{M} = \realn$ (for simplicity) and 
consider the Riemannian structure $\langle \delta x_1, \delta x_2\rangle_x :=\delta x_1^TP(x)\delta x_2$ 
for each $x\in\mathcal{M}$ and each $\delta x_1, \delta x_2\in T_x\mathcal{M}$, 
where $P(x)$ is a symmetric and positive definite matrix in $\real^{n\times n}$ for each $x\in \mathcal{M}$. 
Then, the function $V: T\mathcal{M}\to \real_{\geq 0}$  
given by $V(x,\delta x) := \langle \delta x, \delta x\rangle_x$ 
satisfies the conditions of Definition \ref{def:LyapFins_function}. Moreover,  
from Definition \ref{def:distance}, the distance induced by { $F=\sqrt{V}$} 
is given by the length of the geodesic connecting $x_1$ and $x_2$. 

For the particular selection $P(x) = I$, $V(x,\delta x)$ reduces to $|\delta x|^2_2$. 
Thus,
$d(x_1,x_2) = \int_0^1 \left|\frac{\partial\gamma(s)}{\partial s}\right|_2ds$
where ${\gamma}$ is the straight line $\gamma(s):=(1-s)x_1 + s x_2$.
Therefore, $d(x_1,x_2) = \int_0^1 |x_2-x_1|_2 ds =  |x_1-x_2|_2$.
Note that for distances $d$ given by $k$-norms $d(x_1,x_2) := |x_1-x_2|_k$, where $k \in \mathbb{N}$, 
$k\neq 2$, and $x_1,x_2\in\mathcal{M}$, a quadratic Finsler-Lyapunov function $V$ 
(i.e. $F$ given by a Riemannian structure) is too restrictive. 
Nevertheless, taking $V(x,\delta x) := |\delta x|_k$, we have that $d(x_1,x_2) = |x_1-x_2|_k$.
\end{example}

{ 
\begin{example}
We illustrate 
the importance of the relation \eqref{eq:FinsLyap_bounds}
between Finsler-Lyapunov functions $V$ and Finsler structures $F$. 
As a first example, consider the manifold $\mathcal{M}=\real^2$ and take 
$V(x,\delta x) = 1$ for each $x\in \real^2$ and $\delta x\in\real^2$.
Clearly, a function $F$ that satisfies
Items (i)-(iv) in Definition \ref{def:LyapFins_function} and \eqref{eq:FinsLyap_bounds}
does not exist.
However, mimicking \eqref{eq:induced_distance_main}, we could consider 
the following notion of ``distance'' based on $V$,
$d(x_1,x_2) := \inf_{\Gamma(x_1,x_2)} \int_I V(\gamma(s),\dot{\gamma}(s)) ds$.
Given any to points $x_1,x_2\in\mathcal{M}$, consider a generic curve 
$\gamma:I\subset\real_{\geq 0}\to\mathcal{M}$, $I=[0,1]$, such that $\gamma(0)=x_1$ and $\gamma(1) = x_2$. 
Then,  
$\int_I V(\gamma(s),\dot{\gamma}(s)) = \int_I 1 ds = 1$. 
Consider now a reparameterization of $\gamma$ given by $\gamma_k:I_k\to \mathcal{M}$, $I_k=[0,\frac{1}{k}]$, 
such that $\gamma_k(0)=x_1$ and  $\gamma_k(\frac{1}{k})=x_2$ for any $k>1$.
By definition, we get that 
$d(x_1,x_2) \leq \lim\nolimits\limits_{k\to\infty} \int_{I_k} 1 ds = \lim\nolimits\limits_{k\to\infty} \frac{1}{k} = 0$, 
for any given $x_1,x_2\in\mathcal{M}$. Thus, $d$ is non-negative and 
satisfies the triangle inequality but
$d(x_1,x_2)=0 $ for $x_1\neq x_2$. Therefore, $d$ is not a distance. 
Note that a similar argument extends to 
$V(x,\delta x) = W(x)$ where $W(x)$ is a positive and continuously differentiable function.

As a second example, consider the simplified setting $\mathcal{M}=\real$.
Given the points $0$ and $1$, 
consider the curve $\gamma_k(s): [0,\frac{1}{k}] \to \real$  such that $\gamma_k(s) = ks$,
$k\in \mathbb{N}_{\geq 1}$. The function $V(x,\delta x) := |\delta x|^{p_1} + |\delta x|^{p_2}$
is a candidate Finsler-Lyapunov function only if $p_1=p_2$, with 
Finsler structure $F$ given by $F(x,\delta x) = |\delta x|$.
Otherwise, a function $F$ that satisfies \eqref{eq:FinsLyap_bounds} 
and the homogeneity property in (iii) does not exists.
As above, integrating $V$ does not provide a distance. 
For instance, for any given $p$, and any given $p_1$ and $p_2$, we have that 
$ \int_0^{\frac{1}{k}}V(\gamma_k(s),\dot{\gamma}_k(s))^{\frac{1}{p}} 
= \int_0^{\frac{1}{k}}\left(k^{p_1} + k^{p_2}\right)^{\frac{1}{p}}ds
= \frac{1}{k}\left(k^{p_1} + k^{p_2}\right)^{\frac{1}{p}}$ 
which preserves a constant value for any given reparameterization 
$\gamma_k$ only when $p=p_1=p_2$. 
\end{example}
}
 The reader will notice that 
 the distance $d$ {induced by the Finsler structure $F$ associated to a candidate Finsler-Lyapunov  function \eqref{eq:FinsLyap_bounds}  is not symmetric in general, that is,
 we may have  $d(x,y) \neq d(y,x)$ for some $x,y \in \mathcal{M}$. To induce a symmetric distance,
  it is sufficient to strengthen (iii) in Definition \ref{def:LyapFins_function} to 
 {
 (iii)$_b$ $F(x,\lambda \delta x) = |\lambda| F(x,\delta x)$ for each $\lambda$, 
 and each $(x,\delta x) \in T\mathcal{M}$} (absolute homogeneity, \cite{Tamassy08}). 
 Note that adopting (iii)$_b$ 
 reduces the generality of the class of Finsler-Lyapunov functions excluding, for example, 
 Randers metrics \cite[Section 1.3]{Shen00}.

\section{A Finsler-Lyapunov theorem \\ for contraction analysis}
\label{sec:FinsLyap_thm}

{

Consider a manifold $\mathcal{M}$ of dimension $d$.
In what follows, we exploit the manifold structure of the tangent bundle $T\mathcal{M}$ to provide 
geometric conditions for contraction in local coordinates.
Any given chart $\varphi : \mathcal{U}\subseteq \mathcal{M} \to \real^d$
induces a natural chart on $T\mathcal{U}\subseteq T\mathcal{M}$ (see Section \ref{sec:notation})
that maps each point $(x,\delta x)\in T\mathcal{M}$ to its 
coordinate representation 
$(x_\ell,\delta x_\ell) := (\varphi(x), D\varphi(x)[\delta x]) \in \real^d \times \real^d$.
In local coordinates \eqref{eq:sys} is represented by
$\dot{x}_\ell = f_\ell(t,x_\ell)$ where $f_\ell:\real^d \to \real^d$ is given by
$f_\ell(t,x_\ell) = D\varphi(x)[f(t,x)]$ at $x = \varphi^{-1}(x_\ell)$.
In a similar way, the chart representation $V_\ell(x_\ell,\delta x_\ell)$
of a Finsler-Lyapunov function $V$ is given by
$V(x,\delta x)$ computed at $(x,\delta x) = (\varphi^{-1}(x_\ell), D\varphi^{-1}(x_\ell)[\delta x_\ell])$. 
With a slight abuse of notation, in what follows we drop the subscript $\ell$. 
}

\begin{theorem}
 \label{thm:inf_contraction}
Consider the system \eqref{eq:sys} on a smooth manifold $\mathcal{M}$ with $f$ of class $C^2$,
a connected and forward invariant set $\mathcal{C}$, 
and a function ${\alpha}:\real_{\geq 0} \to \real_{\geq 0}$.
Let $V$ be a candidate Finsler-Lyapunov function such that, in coordinates, 
{
\begin{equation}
 \label{eq:diff_contraction}
  \frac{\partial V(x,\delta x)}{\partial x} f(t,x) + 
  \frac{\partial V(x,\delta x) }{\partial \delta x} \frac{\partial f(t,x)}{\partial x}\delta x \leq -{\alpha}(V(x,\delta x)) 
\end{equation}
}
for each { $t\in \real$}, $x\in \mathcal{C}\subseteq \mathcal{M}$, 
and $\delta x \in T_x \mathcal{M}$. Then, \eqref{eq:sys} is
\begin{description}
 \item[\emph{ (IS)}] incrementally stable on $\mathcal{C}$ if $\alpha(s) = 0$ for each $s\geq 0$;
 \item[\emph{(IAS)}] incrementally asymptotically stable on $\mathcal{C}$ if $\alpha$ is a $\mathcal{K}$ function;
 \item[\emph{(IES)}] incrementally exponentially stable on $\mathcal{C}$ if $\alpha(s) = \lambda s> 0$ for each $s> 0$.
\end{description} \vspace{-5mm}
\end{theorem}
We say that the system \eqref{eq:sys} \emph{contracts} $V$ in 
$\mathcal{C}$ if \eqref{eq:diff_contraction} is satisfied for some function 
$\alpha$ of class $\mathcal{K}$.  $V$ is called the
\emph{contraction measure}, and $\mathcal{C}$ the \emph{contraction region}.

The conditions of the theorem for incremental stability are reminiscent of 
classical Lyapunov conditions for stability, asymptotic stability and exponential 
stability \cite[Chapter 4]{Khalil3},
lifted to the tangent bundle  $T\mathcal{M}$. 
{
 In fact, \eqref{eq:diff_contraction} guarantees that $V$ decreases
along the trajectories of the variational system (in coordinates)
$\dot{x} = f(x)$, $\dot{\delta x} = \frac{\partial f(x)}{\partial x}\delta x$.
The reader will notice that along any solution $\psi_{t_0}(t,x_0)$ to \eqref{eq:sys},
$\dot{\delta x} = \left[\frac{\partial f(x)}{\partial x}_{|x=\psi_{t_0}(t,x_0)} \right]\delta x$ 
characterizes the linearization of \eqref{eq:sys} along its trajectories.
}
Thus, exploiting the relation between $V$ and Finsler structure, 
the contraction of the structure along $\psi_{t_0}(t,x_0)$
(locally - in each tangent space) 
guarantees, via integration, that
the distance between any pair of solutions 
$\psi_{t_0}(t,x_1)$ and $\psi_{t_0}(t,x_2)$, $x_1,x_2\in \mathcal{C}$,
shrinks to zero as $t$ goes to infinity.
A graphical illustration is provided in Fig. \ref{fig:main}.
 \begin{figure}[htbp]
 \begin{center}
 \includegraphics[width=0.98\columnwidth]{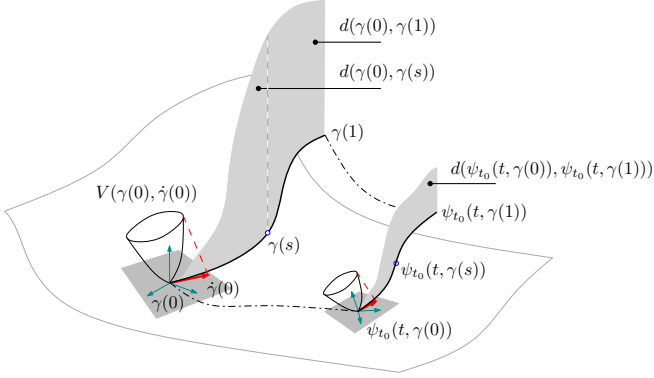} 
 \caption{A graphical illustration of the contraction of the distance 
 induced by Condition \eqref{eq:diff_contraction} on the solutions to \eqref{eq:sys}. The Finsler-Lyapunov function assigns
a positive value to each pair $(\gamma(s),\dot{\gamma}(s))$. The length of the curve $\gamma$ is given by the 
 integral of the Finsler-Lyapunov function along $\gamma$, represented by the shaded area.}
 \label{fig:main}
 \end{center}
 \end{figure} 

The incremental Lyapunov approach 
proposed in \cite{Angeli00}, 
establishes incremental stability by 
checking a pointwise geometric condition in the
product space $\mathcal{M}\times\mathcal{M}$.
In contrast, the differential approach 
proposed here establishes incremental stability
by checking a pointwise geometric condition in the
tangent bundle $T\mathcal{M}$.
Several earlier works have adopted this approach in a Riemannian framework, 
focusing on \emph{quadratic}
functions $V(x,\delta x) =  \delta x^T P(x) \delta x$ in Euclidean spaces 
(see Section VI).
There are a number of reasons to consider Finsler 
generalizations of Riemannian structures for contraction analysis, 
some of which are illustrated  in the next section, where
we report a detailed comparison between the conditions proposed in Theorem 
\ref{thm:inf_contraction} and several results available in literature.

Before entering into the details of the proof, we present a scalar example
that illustrates the value of non-{constant} Riemannian structures in
nonlinear spaces.
\begin{example}
\label{ex:oscillator}
For $\mathcal{M} = \mathbb{S}^1$ consider the dynamics 
\begin{equation}
\label{eq:ex_a1}
\dot{\vartheta} = f(\vartheta) := -\sin(\vartheta).
\end{equation}
The tangent space at every point $\vartheta\in \mathcal{M}$ is given by $\real$.
The naive choice $V_1(\vartheta, \delta \vartheta) := \frac{1}{2}\delta \vartheta^2$
corresponds to a { constant} Riemannian structure on $\mathcal{S}^1$. 
Then, for any given compact set  $\mathcal{C} \subset (-\frac{\pi}{2},\frac{\pi}{2})$ \eqref{eq:diff_contraction} yields
\begin{equation}
\frac{\partial V_1(\vartheta, \delta \vartheta)}{\partial  \delta \vartheta}  \left(\frac{\partial f(\vartheta)}{\partial \vartheta} \right)\delta \vartheta
= - \cos(\vartheta) \delta \vartheta^2 < -\varepsilon V_1(\delta \vartheta) 
\end{equation}
where $\varepsilon > 0$ (sufficiently small).
From Theorem \ref{thm:inf_contraction} we conclude that \eqref{eq:ex_a1}
is incrementally exponentially stable on compact sets $\mathcal{C} \subset (-\frac{\pi}{2},\frac{\pi}{2})$
such that $0\in \mathcal{C}$ (to guarantee that $\mathcal{C}$ is forward invariant).
For $\mathcal{C} = [-\frac{\pi}{2},\frac{\pi}{2}]$ we have only incremental stability, 
since $\cos(\vartheta)=0$ at $|\vartheta|=\frac{\pi}{2}$. From Definition \ref{def:distance}, note that
the distance induced by {$F = \sqrt{2V_1}$} is given by $|\vartheta_1-\vartheta_2|$.

A maximal contracting region is captured with the choice $V_2:(\mathbb{S}^1\setminus\{\pi\})\times \real \to \real_{\geq 0}$
given by $V_2 = \frac{\delta \vartheta^2}{1+\cos{\vartheta}} $. Despite the identification of each
$T_\vartheta\mathcal{M}$ with $\real$, the measure of the ``length'' of $\delta \vartheta$ given by $V_2$ 
now depends { on $\vartheta$}. 
Note that $V_2$ satisfies each condition of Definition \ref{def:LyapFins_function}
and is well defined in $\mathbb{S}^1\setminus\{\pi\}$ since $\frac{1}{1+\cos(\vartheta)}\to \infty$ 
as $|\vartheta|\to \pi$. For any given compact set $\mathcal{C} \subset (\mathbb{S}^1\setminus\{\pi\})$ 
such that $0\in \mathcal{C}$, 
\eqref{eq:diff_contraction} yields
\begin{equation}
\begin{array}{l}
\frac{\partial V_2(\vartheta,\delta \vartheta)}{\partial  \vartheta} f(\vartheta) 
+
\frac{\partial V_2(\vartheta,\delta \vartheta)}{\partial  \delta \vartheta}  \left(\frac{\partial f(\vartheta)}{\partial \vartheta} \right)\delta \vartheta = \vspace{1mm}\\
\qquad = -\frac{\sin(\vartheta)^2}{(1+\cos(\vartheta))^2}\delta \vartheta^2
-2 \frac{\cos(\theta)}{1+\cos(\theta)}\delta\vartheta^2 \vspace{2mm}\\
\qquad = 
-\frac{\sin(\vartheta)^2 +2\cos(\vartheta)(1+ \cos(\vartheta))}{(1+\cos(\vartheta))^2}
\delta\vartheta^2 \vspace{2mm}\\
\qquad = 
-\frac{1 +  2\cos(\vartheta) + \cos(\vartheta)^2}{(1+\cos(\vartheta))^2}
\delta\vartheta^2 \vspace{2mm}\\
\qquad =  - \delta\vartheta^2 \vspace{2mm}\\
\qquad \leq 
-\varepsilon V_2(\vartheta,\delta \vartheta), 
\end{array}
\end{equation}
where $\varepsilon > 0$. Thus, by Theorem \ref{thm:inf_contraction}, 
\eqref{eq:ex_a1} is incrementally exponentially stable on $\mathcal{C}$.
\end{example}

\begin{proofof}\emph{Theorem \ref{thm:inf_contraction}.}
 \label{proof:inf_contraction}
The proof is divided in four main steps. 
For simplicity, we develop the calculations in coordinates.

\emph{(i) Setup: Finsler structure and parameterized solution}.

For any two points $x_1,x_2\in \mathcal{M}$, 
let $\Gamma(x_1,x_2)$ be the collection of
piecewise $C^1$, equally oriented curves $\gamma:I \to \mathcal{C}\subset\mathcal{M}$,
$I:=\{s \in \real \,|\,0 \leq s \leq 1\}$,  
connecting $x_1$ to $x_2$, that is, $\gamma(0) = x_1$ and 
$\gamma(1) = x_2$. { In coordinates, the distance $d$ {induced by $F$} in
Definition \ref{def:distance} reads
\begin{equation}
 \label{eq:induced_distance}
 d(x_1,x_2) = \inf_{\Gamma(x_1,x_2)} \int_I F\left(\gamma(s),\frac{\partial \gamma(s)}{\partial s}\right) ds.
\end{equation}
where $F$ is the associated Finsler structure to $V$ of 
Definition \ref{def:LyapFins_function}.}
For any two initial conditions $x_1,x_2\in \mathcal{C}$ and any given $\varepsilon>0$, 
consider now a regular smooth curve 
$\overline{\gamma}:I\to \mathcal{C}\subset\mathcal{M}$ such that $\overline{\gamma}(0)=x_1$, $\overline{\gamma}(1)=x_2$, and 
\footnote{By using a generic curve smooth $\overline{\gamma}$ which satisfies \eqref{eq:M01} we do not need to assume the existence of geodesics and we simplify the exposition by avoiding the analysis of 
points of non-differentiability.}
{
\begin{equation}
 \label{eq:M01}
  \int_I F\left(\overline{\gamma}(s),\frac{\partial \overline{\gamma}(s)}{\partial s}\right) ds \leq (1+\varepsilon)d(x_1,x_2). 
\end{equation}
}
Let $\psi_{t_0}(\cdot,\overline{\gamma}(s))$ be the solution to \eqref{eq:sys} from the initial condition 
$\overline{\gamma}(s)$, for $s\in I$, at time $t_0$. Precisely, 
$\psi_{t_0}(\cdot,\overline{\gamma}(\cdot))$ is a function from $\real\times I$ to $\mathcal{M}$
{
that satisfies, in coordinates,
\begin{equation}
 \label{eq:M02}
 \frac{\partial }{\partial t} \psi_{t_0}(t,\overline{\gamma}(s)) = f(t,\psi_{t_0}(t,\overline{\gamma}(s))) \quad \forall t\geq t_0\,,\ \forall\in I.
\end{equation}
}
Clearly $\psi_{t_0}(t_0,\overline{\gamma}(\cdot))=\overline{\gamma}(\cdot)$ thus, 
from \eqref{eq:M01}, we have that
{
\begin{equation}
 \label{eq:M03}
  \int_I F\left( \psi_{t_0}(t_0,\overline{\gamma}(s)), 
  \frac{\partial }{\partial s}\psi_{t_0}(t_0,\overline{\gamma}(s)) \right) ds 
  \leq (1+\varepsilon)d(x_1,x_2).
\end{equation}
}
As usual, for each  $t\geq t_0$ and $s\in[0,1]$ the differential of $\psi$ in the direction 
$\frac{\partial}{\partial t}$ characterizes the \emph{time derivative} of the 
parameterized solution $\psi_{t_0}(\cdot,\overline{\gamma}(s))$. 
Instead, the differential of $\psi$ in the direction $\frac{\partial}{\partial s}$
 characterizes at each $s$ 
the \emph{tangent vector} to the curve $\psi_{t_0}(t,\overline{\gamma}(\cdot))$, for fixed  time $t$. 
Following \cite{Lohmiller1998}, we call this tangent vector \emph{virtual displacement}. 
Thus, combining integration of the displacement
along $\frac{\partial}{\partial s}$, time derivative along $\frac{\partial}{\partial t}$, and \eqref{eq:diff_contraction}, we can establish 
contraction of the distance \eqref{eq:induced_distance} along the solutions to \eqref{eq:sys}. 

\emph{(ii)~The displacement dynamics along the solution $\psi_{t_0}(\cdot,\overline{\gamma}(s))$}.

Consider the function $\delta \psi_{t_0}(\cdot,\cdot):\real\times I \to T\mathcal{M}$ given by 
the tangent vector 
{
$\delta \psi_{t_0}(t,s) := D\psi_{t_0}(t,\overline{\gamma}(s))[0,1]$, which
in coordinates is given by $\frac{\partial}{\partial s} \psi_{t_0}(t,\overline{\gamma}(s))$}
for each $t\geq t_0$ and $s\in I$.
{
Its time derivative is given by
\begin{subequations}
\label{eq:M04}
\begin{eqnarray}
 \frac{\partial}{\partial t}\delta \psi_{t_0}(t,\overline{\gamma}(s))
 &=& \frac{\partial^2}{\partial t \partial s}\psi_{t_0}(t,\overline{\gamma}(s)) \label{eq:M04a}\\
 &=& \frac{\partial^2}{\partial s \partial t}\psi_{t_0}(t,\overline{\gamma}(s)) \label{eq:M04b}\\
 &=& \frac{\partial}{\partial s} f(t,\psi_{t_0}(t,\overline{\gamma}(s))) \label{eq:M04c}\\
 &=& \left[ \frac{\partial f(t,x)}{\partial x}\right] 
 \frac{\partial }{\partial s} \psi_{t_0}(t,\overline{\gamma}(s)) \label{eq:M04d}\\
 &=& \left[ \frac{\partial f(t,x)}{\partial x}\right]\delta \psi_{t_0} (t,s). \label{eq:M04e}
\end{eqnarray}
\end{subequations}
where $\frac{\partial f(t,x)}{\partial x}$ must be
evaluated at $x=\psi_{t_0}(t,\overline{\gamma}(s))$.
\eqref{eq:M04a} follows from the definition of $\delta \psi_{t_0}(t,s)$.
\eqref{eq:M04b} follows from the fact that $\psi_{t_0}(\cdot,\overline{\gamma}(\cdot))$ 
is a $C^2$ function, since $f$ is a $C^2$ vector field and $\overline{\gamma}(\cdot)$
is a smooth curve \cite[Theorem 4.1]{boothby2003}.
\eqref{eq:M04d} follows from the chain rule.
Finally, \eqref{eq:M04e} follows from the definition of $\delta \psi_{t_0}(t,s)$.
}

\vspace{1mm}
\emph{(iii)~The dynamics of $V$ along the solution $\psi_{t_0}(\cdot,\overline{\gamma}(s))$}.

Consider the function $\overline{V}: \real\times I \to \real_{\geq 0}$ given by 
$\overline{V}(t,s) = V(\psi_{t_0}(t,\overline{\gamma}(s)), \delta\psi_{t_0}(t,s))$
for each $t\geq t_0$ and $s\in I$. Note that $\overline{V}$ has a well-defined 
time derivative $\frac{d}{dt} \overline{V}(t,s)$ since 
$\overline{V}(t,s)\in\real_{\geq 0}$ for each $t$ and $s$. 
{
In coordinates, for $x=\psi_{t_0}(t,\overline{\gamma}(s))$ and  $\delta x = \delta \psi_{t_0}(t,s)$,
\begin{subequations}
\label{eq:M05}
\begin{eqnarray}
\hspace{-6mm}\frac{d}{dt} \overline{V}(t,s) 
&=& 
\left[\frac{\partial V(x,\delta x)}{\partial x}\right] \frac{\partial}{\partial t} \psi_{t_0}(t,\overline{\gamma}(s)) + \nonumber \\
& & +  
\left[\frac{\partial V(x,\delta x) }{\partial \delta x}\right]
\frac{\partial}{\partial t} \delta \psi_{t_0}(t,s) \label{eq:M05a}\\
&=& 
\left[\frac{\partial V(x,\delta x)}{\partial x}\right] f(t,\psi_{t_0}(t,\overline{\gamma}(s)))  + \nonumber \\
& & +  
\left[\frac{\partial V(x,\delta x) }{\partial \delta x}\right]
\left[ \frac{\partial f(t,x)}{\partial x}\right]\delta \psi_{t_0} (t,s) \label{eq:M05b}\\
&\leq&
-\alpha(\overline{V}(t,s)). \label{eq:M05c} 
\end{eqnarray}
\end{subequations}
 \eqref{eq:M05a} follows from the application of the chain rule. 
 \eqref{eq:M05b} follows from \eqref{eq:M02} and \eqref{eq:M04}.
\eqref{eq:M05c} is enforced by \eqref{eq:diff_contraction}. 
}

\emph{(iv)~Incremental stability properties}.
{
Consider the Finsler structure $F$ associated to the Finsler-Lyapunov function $V$.
Define $\overline{F}: \real\times I \to \real_{\geq 0}$ as
$\overline{F}(t,s) = F(\psi_{t_0}(t,\overline{\gamma}(s)), \delta\psi_{t_0}(t,s))$.}

\textbf{(IS)}~Incremental stability: if $\alpha(s) = 0$ for each $s>0$ then 
\begin{equation}
\label{eq:MStab}
\overline{V}(t,s) \leq \overline{V}(t_0,s) \quad\mbox{for all } t\geq t_0 \mbox{ and } s\in I.
\end{equation}
Therefore, for each $t\geq t_0$, exploiting \eqref{eq:FinsLyap_bounds} and \eqref{eq:MStab}, we get
{
\begin{equation}
\label{eq:dMStab}
\begin{array}{rcl}
d( \psi_{t_0}(t,x_1) , \psi_{t_0}(t,x_2) ) 
&\leq &  
\int_I \overline{F}(t,s) ds \\
&\leq & c_1^{-\frac{1}{p}} \int_I \overline{V}(t,s)^{\frac{1}{p}} ds  \\
&\leq &  c_1^{-\frac{1}{p}} \int_I \overline{V}(t_0,s)^{\frac{1}{p}} ds  \\
& \leq &
(c_2/c_1)^{\frac{1}{p}} \int_I \overline{F}(t_0,s) ds \\
&\leq & (1+\varepsilon) (c_2/c_1)^{\frac{1}{p}} d(x_1,x_2)
\end{array}
\end{equation}
}
where the first inequality follows from the definition of induced distance in 
\eqref{eq:induced_distance_main}, and the last inequality follows from \eqref{eq:M03}.

\textbf{(IAS)}~Incremental asymptotic stability: if $\alpha$ is a 
$\mathcal{K}$ function then {$\frac{d}{dt} \overline{V}(t,s)\leq 0$, thus} (IS) holds,
moreover
by {\cite[Lemma 6.1]{Sontag1989} and \cite[Theorem 6.1]{Hale1980}}, 
there exists a $\mathcal{KL}$ function $\beta$ such that 
\begin{equation}
\label{eq:MAStab}
\overline{V}(t,s) \leq \beta(\overline{V}(t_0,s), t-t_0) \quad\mbox{for all } t\geq t_0 \mbox{ and } s\in I.
\end{equation}
Therefore, following the calculations in \eqref{eq:dMStab}, for each $t\geq t_0$,
\begin{equation}
\label{eq:dMAStab1}
\begin{array}{rcl}
d( \psi_{t_0}(t,x_1) , \psi_{t_0}(t,x_2) ) 
&\leq & {  c_1^{-\frac{1}{p}}} \int_I \overline{V}(t,s)^{\frac{1}{p}} ds \\
&\leq & {  c_1^{-\frac{1}{p}}} \int_I \beta(\overline{V}(t_0,s), t-t_0)^{\frac{1}{p}} ds
\end{array}
\end{equation}
from which we get
\begin{equation}
\label{eq:dMAStab2}
\begin{array}{l}
 \lim\nolimits\limits_{t\to\infty} d( \psi_{t_0}(t,x_1) , \psi_{t_0}(t,x_2) ) \\
\qquad \leq {  c_1^{-\frac{1}{p}}} \lim\nolimits\limits_{t\to\infty} \int_I 
\beta(\overline{V}(t_0,s), t-t_0)^{\frac{1}{p}} ds \\
\qquad = 0 .
\end{array}
\end{equation}
{ The last identity is a consequence of the Lebesgue's dominated convergence
theorem, since $\beta(\overline{V}(t_0,s), t-t_0)$ is a monotonically decreasing function
for $t \to \infty$.
}

\textbf{(IES)}~Incremental exponential stability:
if $\alpha(s) = \lambda s >0 $ for each $s>0$ then, 
by \cite[Theorem 6.1]{Hale1980}, we get
\begin{equation}
\label{eq:MEStab}
\overline{V}(t,s) \leq e^{ -\lambda (t-t_0)}\overline{V}(t_0,s) \quad\mbox{for all } t\geq t_0 \mbox{ and } s\in I.
\end{equation}
Therefore, mimicking \eqref{eq:dMStab}, for each $t\geq t_0$,
\begin{equation}
\label{eq:dMEStab}
\begin{array}{l}
d( \psi_{t_0}(t,x_1) , \psi_{t_0}(t,x_2) ) \\
\qquad \leq   
{
c_1^{-\frac{1}{p}} \int_I \overline{V}(t,s)^{\frac{1}{p}} ds} \\
\qquad \leq   
{
c_1^{-\frac{1}{p}} e^{-\frac{\lambda}{p}(t-t_0)}\int_I \overline{V}(t_0,s)^{\frac{1}{p}} ds 
} \\
\qquad \leq   
{ 
(c_2/c_1)^{\frac{1}{p}} e^{-\frac{\lambda}{p}(t-t_0)} \int_I \overline{F}(t_0,s) ds 
}
\\
\qquad \leq   
{ 
 (1+\varepsilon) (c_2/c_1)^{\frac{1}{p}} e^{-\frac{\lambda}{p}(t-t_0)} d(x_1,x_2).
}
\end{array} \vspace{-5mm}
\end{equation}
\end{proofof}

The proof of Theorem \ref{thm:inf_contraction} generalizes the argument
proposed in the proof of \cite[Lemma 1]{Sontag10yamamoto} and \cite[Theorem 5]{Russo10} 
to general manifolds and Finsler structures (the proof provided in \cite{Russo10}
is developed for Euclidean spaces using matrix measures). An equivalent proof to Theorem 
\ref{thm:inf_contraction} for incremental exponential stability and  $V$ restricted to 
Riemannian structures can be found in \cite[Appendix II]{Rouchon2003}. 
\begin{remark}
 {
 Consider the case $V(x,\delta x) = F(x,\delta x)^p$ in Definition \ref{def:LyapFins_function}.}
 Then, from \eqref{eq:M01} and \eqref{eq:dMStab}, for any given converging sequence 
 $\varepsilon_k\in\real_{> 0}$, $\lim\nolimits\limits_{k\to\infty} \varepsilon_k=0$, 
 we can construct a sequence of $C^2$ curves $\gamma_{k}:I_{k}\to\mathcal{M}$ such that 
 \begin{equation}
 \begin{array}{l}
 \lim\nolimits\limits_{k\to\infty} \int_{I_{k}}  V(\gamma_{k}(s),D\gamma_{k}(s)[1])^{\frac{1}{p}} ds \\
 \qquad \qquad \leq \lim_{k\to\infty}(1+\varepsilon_k)d(x_1,x_2) \\
 \qquad \qquad= d(x_1,x_2). 
 \end{array}
 \end{equation}
 In such a case, in the limit of $k\to\infty$, (IS) in Theorem 
 \ref{thm:inf_contraction} guarantees incremental stability with the stronger property that  
 \begin{equation}
   d( \psi_{t_0}(t,x_1) , \psi_{t_0}(t,x_2) ) \leq d(x_1,x_2) 
   \quad \forall t\geq t_0, \forall x_1,x_2\in\mathcal{M}. \vspace{-5mm}
 \end{equation}
\end{remark}

\begin{remark}
\label{rem:piecewise_diff}
The result of Theorem \ref{thm:inf_contraction} can be extended to
piecewise continuously differentiable and locally Lipschitz candidate Finsler-Lyapunov functions $V$.
In a similar way, the assumption that every two points of $\mathcal{C}$
are connected by a smooth curve $\gamma:I\to \mathcal{C}$ can be relaxed to
piecewise smooth curves.
The key observation is that the decrease of the distance between any two solutions 
is preserved also if \eqref{eq:M05} holds for \emph{almost every} $t$ and $s$. 
With this aim, for example, let $\mathcal{D}\subseteq {T\mathcal{M}}$ be the
set of nondifferentiable points of $V$.
\eqref{eq:M05} holds for \emph{almost every} $t$ and $s$ if
for any given solution $\psi_{t_0}$ such that 
$(\psi_{t_0}(t,x), D \psi_{t_0}(t,x)[0,\delta x]) \in \mathcal{D}$,
there exists  $\varepsilon>0$ which guarantees
$(\psi_{t_0}(\tau,x), D \psi_{t_0}(\tau,x)[0,\delta x]) \notin \mathcal{D}$ 
for every $\tau\in(t,t+\varepsilon]$.
The transversality of the trajectories with respect to $\mathcal{D}$ 
can be enforced geometrically by requiring that, (in coordinates)
for each $t\geq t_0$, and each 
$(x,\delta x)\in\mathcal{D}$,
the pair $(f(t,x), \frac{\partial}{\partial x}f(t,x)\delta x)$ does not belong to the 
tangent cone to $\mathcal{D}$ at $(x,\delta x)$.
\end{remark}

We conclude the section by emphasizing the analogy between 
classical Lyapunov theory and Theorem \ref{thm:inf_contraction}.
We also emphasize the geometric (or coordinate-free) nature of 
Theorem \ref{thm:inf_contraction}, showing that
\eqref{eq:diff_contraction} in Theorem 
\ref{thm:inf_contraction} is independent on the
selected coordinate chart. With this aim,
we introduce two charts $\varphi,\psi:\mathcal{U}\subseteq\mathcal{M} \to \real^d$,
and we denote by $z$ and $y$ the coordinate representations 
$z = \varphi(x)$ and $y=\psi(x)$ of any point $x\in \mathcal{M}$.
In particular, $V^{(z)}$ and $f^{(z)}(t,z)$ denote respectively
the Finsler-Lyapunov function $V$ and
the vector field \eqref{eq:sys} in the chart $\varphi$.
$V^{(y)}$ and $f^{(y)}(t,y)$ denote the same quantities in the local chart $\psi$.  

The analogy with classical Lyapunov theory is emphasized by considering the
aggregate state $Z:=(z,\delta z)$. 
Suppose that \eqref{eq:diff_contraction} has been established by using 
the coordinate chart $\varphi$. Exploiting the notion
of aggregate state, we define
$\dot{Z} = f^{(Z)}(Z)$, where
$f^{(Z)}(Z) := \smallmat{f^{(z)}(z) \\  
\frac{\partial f^{(z)}(z)}{\partial z} \delta z}$,
and $V^{(Z)}(Z) := V^{(z)}(z,\delta z)$,
from which \eqref{eq:diff_contraction} reads
$
 \frac{\partial V^{(Z)}(Z)}{\partial Z} f^{(Z)}(Z) \leq -\alpha(V^{(Z)}(Z)).
$
This formulation reveals that the Finsler-Lyapunov approach is 
Lyapunov's second method on the variational system.
Clearly, a Finsler-Lyapunov function differs from classical Lyapunov functions,
since its definition is tailored to endow $\mathcal{M}$ with the structure of a metric space.

Coordinate independence can be shown as follows.
Define $Y:=(y,\delta y)$ and note that
$Z = H(Y)$, where $H(y,\delta y) := (\varphi(\psi^{-1}(y)), 
\frac{\partial \varphi(\psi^{-1}(y))}{\partial y} \delta y)$.
Necessarily, the vector field in the $Y$ coordinates reads
$f^{(Y)}(Y) = 
\left[\frac{\partial H^{-1}(Z)}{\partial Z}_{|Z = H(Y)} \right]
f^{(Z)}(H(Y))$, and $V^{(Y)}(Y) = V^{(Z)}(H(Y))$. Thus,
\begin{equation}
\label{eq:Lyap_analogy2}
\begin{array}{rcl}
 \frac{\partial V^{(Y)}(Y)}{\partial Y} f^{(Y)}(Y) 
 &=& \left[\frac{\partial V^{(Z)}(Z)}{\partial Z}_{|Z=H(Y)}\right] \cdot \vspace{1mm}\\
 & & \cdot \underbrace{\frac{\partial H(Y)}{\partial Y} 
 \left[\frac{\partial H^{-1}(Z)}{\partial Z}_{|Z = H(Y)} \right]}_{=I} \cdot \\
 & &
 \cdot  f^{(Z)}(H(Y)) \vspace{1mm}\\
 &=& \left[\frac{\partial V^{(Z)}(Z)}{\partial Z}_{|Z=H(Y)}\right] f^{(Z)}(H(Y)) \vspace{1mm}\\
 &\leq& -\alpha(V^{(Z)}(H(Y))) \vspace{1mm}\\
 &=& -\alpha(V^{(Y)}(Y)),\\
 \end{array}
\end{equation}
which proves the coordinate independence of \eqref{eq:diff_contraction}.

\section{Revisiting some literature on contraction}
\label{sec:literature_comparison}

\subsection{Riemannian contraction, matrix measure contraction, and incremental stability}
\label{sec:comparison_sontag-slotine}
For a historical perspective on contraction the reader is referred to
\cite{Jouffroy05}, and related concepts in \cite{Pavlov04} and \cite{Sontag10yamamoto}.
We propose here a detailed comparison with selected references from the literature.
First, we consider results on contraction based on matrix measures
\cite{Russo10,Sontag10yamamoto} and matrix inequalities \cite{Pavlov05}.
We recast these results within the differential framework proposed in 
Theorem \ref{thm:inf_contraction}, by
suitable definitions of state-independent Finsler-Lyapunov functions $V(x,\delta x)$.
Then, we consider results based on Riemannian structures \cite{Lohmiller1998,Rouchon2003},
and we show that they coincide with the (IES) condition of Theorem \ref{thm:inf_contraction}
for a function $V(x,\delta x)$ defined by the Riemannian structure. 

The reader will notice that these two groups of results are essentially disjoint. 
The equivalence between the conditions based 
on matrix measures and the conditions based on Riemannian structures 
can be established only 
for quadratic vector norms $|x|_P = \sqrt{x^T P x}$ or, 
equivalently, for state-independent Riemannian structures 
$\langle \delta x, \delta x \rangle =  \delta x^T P \delta x$.
However, both groups of results fall within the proposed
differential Finsler-Lyapunov framework.
{
We emphasize that the early work of Lewis \cite{Lewis1949} already
exploits Finsler structures for the characterization of
incremental properties of solutions, also providing
early results on the relation between contraction and 
the existence of periodic solutions.}

The approach proposed in \cite{Russo10} and \cite{Sontag10yamamoto} is
based on the matrix measure of the Jacobian $J(t,x):=\frac{\partial f(t,x)}{ \partial x}$.
For instance, given a vector norm $|\cdot|$ in $\realn$ and its induced matrix norm, 
the \emph{induced matrix measure} $\mu$ of a matrix $A\in\real^{n\times n}$
is given by  $\mu(A):= \lim\nolimits\limits_{h\to 0^+} \dfrac{|I+hA|-1}{h}$, \cite[Section 3.2]{Vidyasagar}. 
Then, following
\cite[Definition 1 and Theorem 1]{Sontag10yamamoto}, let $\mathcal{C}$ be a convex set,
forward invariant for the system $\dot{x}=f(t,x)$. $f$ is a $C^1$ function.
If 
\begin{equation}
\label{eq:sontag01}
\mu( J(t,x) ) \leq -c < 0 \qquad \mbox{ for each } x\in\mathcal{C} \mbox{  and each } t\geq 0,
\end{equation}
then the system is incrementally exponentially stable with a distance given by 
$d(x_1,x_2) = |x_1-x_2|$. Moreover, by \cite[Lemma 4]{Sontag10yamamoto}, the same
result hold for non convex sets $\mathcal{C}$ that satisfy a mild regularity assumption, 
and it guarantees incremental exponential stability with a distance function 
$d(x_1,x_2) \leq K|x_1-x_2|$ for some $K>1$.

Condition \eqref{eq:sontag01}
{guarantees that} \eqref{eq:diff_contraction} {holds}
for the Finsler-Lyapunov function given by
$V(x,\delta x) = |\delta x|$ and $\alpha(s) = cs$. This follows from
\begin{equation}
\label{eq:sontag02}
\begin{array}{l}
\dfrac{\partial V(x,\delta x)}{ \partial  \delta x} J(t,x) \delta x  \ = \\
\qquad = \lim\nolimits\limits_{h\to 0^+} \dfrac{V(x,\delta x + hJ(t,x)\delta x) - V(x,\delta x)}{h} \\
\qquad \leq \lim\nolimits\limits_{h\to 0^+} \dfrac{|I+hJ(t,x)||\delta x| - |\delta x|}{h} \\
\qquad = \lim\nolimits\limits_{h\to 0^+} \dfrac{|I+hJ(t,x)| - 1}{h} V(x,\delta x)\\
\qquad = \mu(J(t,x))V(x,\delta x)\\
\qquad = -c V(x,\delta x) \qquad \mbox{ for each } t\geq 0,\,x\in\mathcal{C},\,\delta x\in\realn.  
\end{array}
\end{equation}

The approach proposed in \cite{Pavlov05} 
(and in \cite[Chapter 5, Section 5]{willems1970stability} for time-invariant systems)
use matrix inequalities based on 
the Jacobian $J(t,x)$ and on two positive definite and symmetric matrices $P$ and $Q$. 
{ These results are a particular case of the approach based on matrix
measures, for suitable selections of the norm $|\cdot|_2$. It is instructive
to show the equivalence between} \cite[Theorem 1]{Pavlov05} 
and incremental exponential stability of Theorem \ref{thm:inf_contraction} 
for $V$ restricted to the {constant} Riemannian structure $\delta x^T P \delta x$. 
Consider the system 
$\dot{x} = f(x,w(t))$ where $f$ is a $\mathcal{C}^1$ function and 
$w:\real_{\geq 0}\to \mathcal{W}\subset \real^m$ is a $C^1$ exogenous signal. 
Thus, $f(x,w(t))$ is a time-varying $C^1$ function. 
Applying Theorem \ref{thm:inf_contraction} to $V(x,\delta x) = \delta x^T P \delta x$, 
incremental exponential stability holds if
\begin{equation}
\label{eq:Pavlov01}
\begin{array}{rcl}
 \frac{\partial V(x,\delta x)}{ \partial  \delta x} \frac{\partial f(x,w)}{ \partial x} \delta x
 &=& \delta x^T \left(P\frac{\partial f(x,w)}{ \partial x} + \frac{\partial f(x,w)}{ \partial x}^T P \right)\delta x\\
 & \leq & -\lambda V(x,\delta x) = - \lambda \delta x^T P \delta x  
 \end{array}
\end{equation}
for some $\lambda > 0$ and for every $\delta x\in\realn$ and $w\in \mathcal{W}$. 
The right-hand side of \eqref{eq:Pavlov01} can be replaced by $-\delta x^T Q \delta x$, for some matrix $Q=Q^T>0$
(for any given $Q$, we can always find $\lambda$ sufficiently small to guarantee
$Q > \lambda P$, and vice versa). Therefore, the condition in \eqref{eq:Pavlov01} is equivalent to
the existence of  positive definite and symmetric matrices $P$ and $Q$ such that 
\begin{equation}
\label{eq:Pavlov02}
P \frac{\partial f(x,w)}{ \partial x} + \frac{\partial f(x,w)}{ \partial x}^T P \leq -Q
\end{equation}
which is \cite[Eq. (8), Theorem 1]{Pavlov05}. 
The induced distance {given by $F=\sqrt{V}$} is the quadratic form  
$d(x_1,x_2) \!=\! \sqrt{(x_1\!-\!x_2)^T\!P(x_1\!-\!x_2)}$.
See also \cite{Pavlov05book} and Section \ref{sec:partial_contraction} in the present paper.

Conditions for contraction based on quadratic structures  
$\delta x^T M(x) \delta x$ are provided in the contraction 
paper \cite{Lohmiller1998}
(we consider the time-invariant case only). 
\cite[Definition 2 and Theorem 2]{Lohmiller1998}
establish incremental exponential stability for $\dot{x}= f(t,x)$ 
by requiring, using the notation of \cite{Lohmiller1998}, that the inequality
\begin{equation} 
\label{eq:slot01}
\begin{array}{l}
 \delta x^T \left( J(t,x)^T M(x) + M(x)J(t,x) + \dot{M}(x) \right) \delta x \\
 \hspace{4.5cm} \leq -\lambda \delta x^T M(x) \delta x
 \end{array}
\end{equation}
is satisfied for every $x$ and $\delta x$, for some $\lambda >0$. 
Note that $\delta x^T \dot{M}(x) \delta x$ is a short notation for
$\frac{\partial }{\partial x}(\delta x^T M(x) \delta x) f(x)$.
Therefore, taking $V(x,\delta x) = \delta x^T M(x)\delta x$, 
the relation between \eqref{eq:slot01} and \eqref{eq:diff_contraction} for 
incremental exponential stability
is immediate. The same argument illustrates 
the relation between the differential approach proposed here 
and the results in \cite[Appendix II]{Rouchon2003} and \cite[Definition 2.4 and Theorem 2.5]{Zamani11}
(for this last paper, the differential equation $\dot{x} = f(x,u)$, where $u$ is an input signal, 
is casted to the form \eqref{eq:sys} by considering the time-varying vector field $\overline{f}(t,x) := f(x,u(t))$).

We conclude the section by considering the incremental Lyapunov approach in 
\cite{Angeli00,Ruffer11}. The key observation is given by 
\cite[Lemma 2.3 and Remark 2.4]{Angeli00} and \cite[Appendix A.1]{Ruffer11} which 
shows the equivalence between the incremental stability of $\dot{x} =  f(t,x)$, $x\in \realn$,
and the stability of the set $\mathcal{A} := \{(x_1,x_2)\in \real^{2n}\,|\, x_1=x_2\}$ 
for the extended system $\dot{x}_1 =  f(t,x_1)$,
$\dot{x}_2 =  f(t,x_2)$. Thus, to show asymptotic stability of the set $\mathcal{A}$,
a Lyapunov function $V(x_1,x_2)$ must be positive everywhere but on $\mathcal{A}$,
that is 
\begin{equation}
\label{eq:angeli01}
\underline{\alpha}(|x_1-x_2|) \leq V(x_1,x_2) \leq \overline{\alpha}(|x_1-x_2|),
\end{equation}
for some $\underline{\alpha},\overline{\alpha}\in \mathcal{K}$; and the
derivative of $V(x_1,x_2)$ along the solutions of the system must decrease for $x_1,x_2\notin \mathcal{A}$, 
which is established by enforcing 
\begin{equation}
\label{eq:angeli02}
\frac{\partial V(x_1,x_2)}{\partial x_1}f(t,x_1) + 
\frac{\partial V(x_1,x_2)}{\partial x_2}f(t,x_2)
\leq -\alpha(|x_1-x_2|)
\end{equation}
for each pair $x_1,x_2\in\realn$, where $\alpha\in\mathcal{K}$. 
Indeed, an incremental Lyapunov function is essentially a 
Lyapunov function for the extended system which measures directly 
the distance between any two points $x_1$ and $x_2$. 

The differential 
framework proposed here does not use a Lyapunov function 
to study directly the time evolution of the distance between any two solutions. 
Instead, a lifted Lyapunov function on the tangent bundle  is used to characterize the contraction of the 
infinitesimal neighborhood of each point $x$ - a local property - to infer indirectly the contraction
of the distance - a global property -  via integration. Applications suggest that it can be considerably
more difficult to construct a distance than the associated differential structure.

\subsection{Contractive systems forget initial conditions}
\label{sec:partial_contraction}
Under standard completeness assumptions on the distance,
all the (bounded) solutions of a contractive system converge to a unique steady-state solution.
This feature is exploited in control design \cite{Wang05, Pavlov05, Pavlov05book,Jouffroy10},
for example in tracking, by inducing an attractive desired steady-state solution via the feedforward action 
of exogenous signals (that preserve the contraction property),  
or in observer design, by a suitable injection of the measured output.
In what follows we revisit these results, showing that a particular 
application of Theorem \ref{thm:inf_contraction} entails the sufficient conditions
for convergent systems in \cite{Pavlov05, Pavlov05book}, and we 
formulate a proposition whose conditions parallels
the relaxed contraction analysis proposed 
by \cite{Wang05,Jouffroy10}, through the notion of virtual system.

Following \cite{Pavlov05} and \cite{Pavlov05book}, consider the system $\dot{x} = f(x,w(t))$
where $w$ is an exogenous signal. Define $\hat{f}(t,x) := f(x,w(t))$, 
assume that the solutions are bounded, and 
suppose that Theorem \ref{thm:inf_contraction}
holds for $\dot{x} = \hat{f}(t,x)$. Then, by incremental asymptotic stability, 
the solutions of the system converge towards each other, thus every 
solution converges to a steady
state solution $\dot{x}^*(t) = f(x^*(t),w(t))$ induced by $w$. 
This results parallels \cite[Property 3]{Pavlov05}. In particular, 
Theorem \ref{thm:inf_contraction} applied to $\dot{x}=\hat{f}(t,x) = f(x,w(t))$ 
recovers \cite[Property 3]{Pavlov05} when $V=\delta z^T P \delta z$ ({constant} metric) and $\alpha(s) = -ks$, $k>0$. 

Following \cite{Wang05} and \cite{Jouffroy10}, consider the system \eqref{eq:sys} given by $\dot{x} = f(t,x)$ and a new system of equations 
 \begin{equation}
 \label{eq:sys_virtual}
  \dot{z} = \hat{f}(t,z,x) \quad  \mbox{such that}\quad \hat{f}(t,x,x) =f(t,x),\,\hat{f}\in C^1.
 \end{equation}
  \eqref{eq:sys_virtual} is the so-called \emph{virtual system}, \cite{Wang05}.  
 \eqref{eq:sys_virtual} arises naturally in tracking and state estimation problems where, possibly,
 \eqref{eq:sys} is the reference system and the controlled/observer system is given by \eqref{eq:sys_virtual}. 
 For example, $\hat{f}(t,z,x) = {f}(t,z) + K(z-x)$ may represent a tracking controlled system 
 with state-feedback $K(z-x)$, while $\hat{f}(t,z,x) = {f}(t,z) + L(y_z-y_x)$ may represent
 an observer dynamics with output injection $L(y_z-y_x)$.
 Inspired by \cite{Wang05} and \cite{Jouffroy10},
 we provide the following proposition, a straightforward application of Theorem \ref{thm:inf_contraction}. 

\begin{proposition}
\label{prop:partial_contraction}
  Consider the system \eqref{eq:sys} on a smooth manifold $\mathcal{M}$ with $f$ of class $C^2$, and 
  a connected and forward invariant set $\mathcal{C}_x\subseteq\mathcal{M}$ for \eqref{eq:sys}.
  Consider \eqref{eq:sys_virtual} and suppose that the set $\mathcal{C}_z\subseteq\mathcal{M}$ is connected and 
  forward invariant for \eqref{eq:sys_virtual}. Given a $\mathcal{K}$ function ${\alpha}$,
 let $V$ be a candidate Finsler-Lyapunov function for \eqref{eq:sys_virtual} (Definition \ref{def:LyapFins_function}) 
 such that, in coordinates,  
\begin{equation}
\label{eq:partial_contraction}
  \frac{\partial V(z,\delta z)}{\partial z} \hat{f}(t,z,x) + 
  \frac{\partial V(z,\delta z) }{\partial \delta z} \frac{\hat{f}(t,z,x)}{\partial z}\delta z 
  \leq -{\alpha}(V(z,\delta z))
\end{equation}
for each {$t\in\real$}, each $x\in\mathcal{C}_x$ (uniformly in $x$), 
each $z\in \mathcal{C}_z\subseteq \mathcal{M}$, 
and each $\delta z \in T_z \mathcal{M}$. Then,
for any given initial condition $x_0\in \mathcal{C}_x$, and any initial condition $z_0\in\mathcal{C}_z$, 
 each solution $\varphi_{t_0}^{z}(t,z_0)$ to \eqref{eq:sys_virtual} 
 converges asymptotically to the solution $\varphi_{t_0}^{x}(t,x_0)$ to \eqref{eq:sys}.
\end{proposition}

Combining the virtual system decomposition \eqref{eq:sys_virtual} with 
Proposition \ref{prop:partial_contraction} is useful for applications like tracking and 
state estimation, but also as an analysis tool. In fact, 
if Proposition \ref{prop:partial_contraction} holds and \eqref{eq:sys_virtual} converges 
to a given steady-state solution $z^*$ uniformly in $x$, then all solutions of \eqref{eq:sys} converge to that solution.
The conclusion of Proposition \ref{prop:partial_contraction} is a consequence of Theorem \ref{thm:inf_contraction}: 
considering the solution $\varphi^x_{t_0}(t,x_0)$ to \eqref{eq:sys} from a given initial condition $x_0\in\mathcal{C}_x$,
the dynamics \eqref{eq:sys_virtual} can be rewritten as the time-varying dynamics 
$\dot{z} = \tilde{f}(t,z) := \hat{f}(t,z,\varphi^x_{t_0}(t,x_0))$, and 
\eqref{eq:partial_contraction} guarantees that 
the conditions for incremental asymptotic stability of Theorem \ref{thm:inf_contraction} applied to $\dot{z} = \tilde{f}(t,z)$ are satisfied.
Therefore, for any given initial conditions $z_1,z_2$, the solutions $\psi^z_{t_0}(t,z_1)$ and $\psi^z_{t_0}(t,z_2)$
converge towards each other, that is, $\lim\nolimits\limits_{t\to \infty} d(\psi^z_{t_0}(t,z_1),\psi^z_{t_0}(t,z_2)) = 0$. 
The conclusion of the proposition follows by noticing that when $z_2=x_0$, we have that $\psi^z_{t_0}(t,z_2) = \psi^x_{t_0}(t,x_0)$
(since $\hat{f}(t,x,x) = f(t,x)$). Thus, from every initial condition $z_1\in\mathcal{C}_z$,
$\lim\nolimits\limits_{t\to \infty} d(\psi^z_{t_0}(t,z_1),\psi^x_{t_0}(t,x_0)) = 0$. 
Similar conditions are provided in \cite{Wang05} and \cite{Jouffroy10} for Riemannian metrics $V(z,\delta z) = \delta z^T P(z) \delta z$.

\section{LaSalle-like relaxations}
\label{sec:LaSalle}

A very first step of Lyapunov theory is to relax the strict decay of Lyapunov functions
by exploiting the invariance of limit sets. We show that this important relaxation readily 
extends to Finsler-Lyapunov functions.
We only develop the analysis for the particular case of time-invariant differential equations $\dot{x}=f(x)$.

\begin{theorem}[\emph{LaSalle invariance principle for contraction}]
 \label{thm:LaSalle_timeinvariant2}
Consider the system $\dot{x}=f(x)$ on a smooth manifold $\mathcal{M}$ with $f$ of class $C^2$,
a continuous function ${\alpha}:T \mathcal{M}\to \real_{\geq 0}$, 
and  a connected set $\mathcal{C}\subset\mathcal{M}$, 
forward invariant for $\dot{x}=f(x)$.
Let $V$ be a candidate Finsler-Lyapunov function such that,
in coordinates, 
\begin{equation}
 \label{eq:LaSalle_timeinvariant2}
  \frac{\partial V(x,\delta x)}{\partial x} f(x) + 
  \frac{\partial V(x,\delta x) }{\partial \delta x} \frac{\partial f(x)}{\partial x}\delta x 
  \leq -{\alpha}(x,\delta x)
\end{equation}
for each $x\in \mathcal{C} \subset \mathcal{M}$, and each $\delta x \in T_x \mathcal{M}$.
{ Then, for any bounded solution of $\dot{x} = f(x)$ from $\mathcal{C}$, the solutions of the variational system 
$\dot{x}=f(x)$, $\dot{\delta x} = \frac{\partial f(x)}{\partial x}\delta x$ converge
to the largest invariant set ${\Delta}$ contained in
\begin{equation}
\label{eq:LaSalle_PI}
\Pi :=\{ (x,\delta x)\in T\mathcal{M} \,|\, \alpha(x, \delta x)=0 \, , \, x\in \mathcal{C} \}.
\end{equation}
If ${\Delta} = \mathcal{C} \times \{0\}$, then 
$\dot{x} = f(x)$ is incrementally asymptotically stable on $\mathcal{C}$}
\footnote{
{Note that \eqref{eq:LaSalle_timeinvariant2} guarantees incremental stability,
thus boundedness of solutions of $\dot{x} = f(x)$ is for free whenever the system
has an equilibrium $x_e$ or a bounded steady-state solution $x^*(t)$
contained in $\mathcal{C}$.
} 
} .
\end{theorem}
\begin{proof}
{
We adapt the proof of the LaSalle invariance theorem \cite{LaSalle1960}
by exploiting the properties of the variational system. For instance,
\textbf{(i)}~consider a bounded solution of $\dot{x}=f(x)$. By incremental stability (from \eqref{eq:LaSalle_timeinvariant2}
and Theorem \ref{thm:inf_contraction}) all the solutions of $\dot{x}=f(x)$ from $\mathcal{C}$ 
are bounded. This guarantees that, for 
any initial condition $\gamma(s)$, $\gamma:I\!\to\!\mathcal{C}$, 
$s\in I $, the displacement $\frac{\partial}{\partial s} \psi(t,\gamma(s))$ (in coordinates)
of the solution $\psi(t,\gamma(s))$ to $\dot{x}=f(x)$ 
is bounded. Therefore, 
any given solution $(x(\cdot), \delta x(\cdot))$ of
the variational system is bounded;
\textbf{(ii)}~because $\mathcal{C}$ is forward invariant and $(x(\cdot), \delta x(\cdot))$ is bounded, 
its positive limit set $L^+$ 
is a nonempty, compact, invariant set \cite[Lemma 4.1]{Khalil3};
\textbf{(iii)}~$V$ is bounded from below by $0$ and satisfies
$\frac{d}{dt} V(x(t),\delta x(t)) \leq 0$ for any given solution
$(x(\cdot),\delta x(\cdot))$ to the variational system. Thus,
$\lim_{t\to\infty} V(x(t),\delta x(t))$ exists and it is given by some value $c\in\real_{\geq 0}$.
The consequence of \textbf{(i)}-\textbf{(iii)} is that any solution $(y(\cdot),\delta y(\cdot))$ to the 
variational system from $(y(0),\delta y(0)) \in L^+$
necessarily satisfies $V(y(t),\delta y(t)) = c$ for any given $t$, which implies
$\frac{d}{dt} V(y(t),\delta y(t)) = \alpha(y(t),\delta y(t)) = 0$ for all $t$. That is,
$L^+ \subseteq \Pi$.

For incremental asymptotic stability, we have to prove
that for any given curve $\gamma:I\!\to\!\mathcal{C}$,
the solutions $\psi(t,\gamma(s))$ to $\dot{x}=f(x)$ for $s\in I$ satisfies
$\lim\nolimits\limits_{t\to \infty}\int_I F( \psi(t,\gamma(s)) , 
\frac{\partial}{\partial s}\psi(t,\gamma(s))) = 0$.
Using \eqref{eq:FinsLyap_bounds},
this is a consequence of the fact that 
$\lim\nolimits\limits_{t\to \infty}V( \psi(t,\!\gamma(s)) , 
\frac{\partial}{\partial s}\psi(t,\gamma(s))) = 
V( \psi(t,\!\gamma(s)) , 0) = 0$, 
for each $s\in I$. Note that the first identity follows from the assumption that
$\mathcal{C}\times\{0\}$ is the largest invariant set contained in $\Pi$.
}
\end{proof}
To the best of authors' knowledge, an invariance principle has not appeared in the literature on contraction. 
This illustrates the potential of a Lyapunov framework for contraction analysis.

We illustrate the use of Theorem \ref{thm:LaSalle_timeinvariant2} in the
following (linear) example, where we take advantage of classical observability
conditions. Example \ref{example:boost} illustrates a general 
class of models in power electronics for which incremental tools are 
frequently used \cite{Ramirez99}.
\begin{example}
\label{example:boost}
Consider the following averaged equations of a single-boost converter \cite{Escobar01}
\begin{equation}
\label{eq:sb_1}
\left\{\begin{array}{rcl}
L \dot{x}_L &=& -u x_C + E \\
C \dot{x}_C &=& u x_L - \frac{1}{R}x_C  
\end{array}\right.
\end{equation}
where $x_L$ is the inductor current, $x_C$ is the capacitor voltage, 
and $E$ is the input voltage. The quantities $L$, $C$, and $R$ are respectively
the inductance, the capacitance and the (load) resistance of the circuit.

We claim that for any given constant input $u^*\neq 0$, and any constant positive value of 
the circuit quantities $L$, $C$ and $R$, the system is incrementally asymptotically
stable. Note that \eqref{eq:sb_1} is a time-invariant linear system for $u=u^*$, 
so that a natural candidate Finsler-Lyapunov function is provided by the incremental energy 
$V(x,\delta x) = \frac{1}{2}(L \delta x_L^2 + C \delta x_C^2)$. In fact,
\begin{equation}
\label{eq:sb_2}
\begin{array}{l}
 {  \frac{\partial V(x,\delta x)}{\partial x} f(x,u^*) + 
  \frac{\partial V(x,\delta x) }{\partial \delta x} \frac{\partial f(x,u^*)}{\partial x}\delta x } = \\
  \qquad = \mymatrix{c}{\delta x_L \\ \delta x_C}^T 
 \mymatrix{cc}{0 & -u^* \\ u^* & -\frac{1}{R}} \mymatrix{c}{\delta x_L \\ \delta x_C} \\
\qquad = - \frac{\delta x_C^2}{R} \leq 0, 
\end{array}
\end{equation}
where $\alpha(x,\delta x) = \frac{\delta x_C^2}{R}$.
By \eqref{eq:LaSalle_PI}, considering $\psi(t,x) = e^{At} x$, we have that
$\Pi_\tau := \{(x,\delta x)\in{\real^2\times\real^2} \,|\, 
\forall t\in[0,\tau], \,
\delta x^T (e^{A t })^T  \smallmat{0 & 0\\ 0 & 1} e^{A \tau} \delta x = 0\}$.
Thus, for any given $\tau > 0$, we have that $\Pi_\tau = \real^2\times\{0\}$.
Incremental asymptotic stability follows from Theorem 2 (from the linear
nature of the system, the incremental asymptotic stability is actually exponential).
\end{example}

\begin{remark}
\label{rem:LaSalle}
For a time-varying differential equation \eqref{eq:sys}, a possible formulation of 
invariance-like conditions for asymptotic stability is given by the inequality,
{
in coordinates, 
\begin{equation}
 \label{eq:LaSalle_int}
  \frac{\partial V}{\partial x} f(t,x) + 
  \frac{\partial V }{\partial \delta x} \frac{\partial f(t,x)}{\partial x}\delta x 
  \leq -{\alpha}(t,x,\delta x)V
\end{equation} }
for each {$ t\in\real$}, $x\in \mathcal{C} \subset \mathcal{M}$, and $\delta x \in T_x \mathcal{M}$, 
where $V$ is a candidate Finsler-Lyapunov and 
${\alpha}:\real_{\geq 0}\times{T\mathcal{M}}\to \real_{\geq 0}$.
Incremental asymptotic stability on $\mathcal{C}$ holds if
\begin{equation}
\label{eq:LaSalle_int_cond}
\lim_{t\to\infty}\int_{t_0}^t \alpha(\tau,\psi_{t_0}(\tau,x),D\psi_{t_0}(\tau,x)[0,\delta x])d\tau = \infty 
\end{equation}
$\mbox{for each } x\in \mathcal{C} \mbox{ and } \delta x \in T_x\mathcal{M}. $
In general, \eqref{eq:LaSalle_int_cond} is established by relying on further analysis of the solutions 
of the system\footnote{The differential $D\psi_{t_0}(t+\tau,x)[0,\cdot]$ assigns to each tangent vector 
$\delta x\in T_{\psi_{t_0}(t,x)}\mathcal{M}$  the tangent vector 
$D\psi_{t_0}(t,x)[0, \delta x] \in T_{\psi_{t_0}(t+\tau,x)}\mathcal{M}$.
Thus, 
$D\psi_{t_0}(t,x)[0, \delta x]$ represents the evolution of the tangent vector
$\delta x$ along the solution $\psi_{t_0}$ after $\tau$ units of time. 
In coordinates, 
$D\psi_{t_0}(t+\tau,x)[0, \delta x] 
= \frac{\partial \psi_{t_0}(t,x)}{\partial x}\delta x$.
}.

By Theorem \ref{thm:inf_contraction}, \eqref{eq:LaSalle_int}
and \eqref{eq:LaSalle_int_cond}
guarantee incremental stability. To see why \eqref{eq:LaSalle_int} and \eqref{eq:LaSalle_int_cond} 
guarantee incremental asymptotic stability, one
has to follow the proof of Theorem \ref{thm:inf_contraction} up to Equation \eqref{eq:M05}, 
by replacing each quantity $\alpha(V(x,\delta x))$ by $\alpha(t,x,\delta x) V(x,\delta x)$. 
From there, using the definition 
 $\overline{\alpha}(t,s) 
 := \alpha(t,\psi_{t_0}(t,\overline{\gamma}(s)),D\psi_{t_0}(t,\overline{\gamma}(s))[0,1])$,
by comparison lemma \cite[Lemma 3.4]{Khalil3} we get
$\overline{V}(t,s) \leq e^{-\int_{t_0}^{t} \overline{\alpha}(\tau,s)d\tau}  \overline{V}(t_0,s)$ 
for all $t\geq t_0$ and $s\in I$, which combined with \eqref{eq:LaSalle_int_cond} guarantees that
\begin{equation}
\label{eq:L04}
\begin{array}{l}
 \lim\nolimits\limits_{t\to\infty} d( \psi_{t_0}(t,x_1) , \psi_{t_0}(t,x_2) ) \\
\qquad \leq  { c_1^{-\frac{1}{p}}}
\lim\nolimits\limits_{t\to\infty} \int_I e^{-\frac{1}{p}\int_{t_0}^{t} \overline{\alpha}(\tau,s)d\tau}  \overline{V}(t_0,s)^{\frac{1}{p}} ds \\
\qquad \leq { c_1^{-\frac{1}{p}}}
\Big(\max\nolimits\limits_{s\in I} \overline{V}(t_0,s)^{\frac{1}{p}}\Big) \lim\nolimits\limits_{t\to\infty} \int_I e^{-\frac{1}{p}\int_{t_0}^{t} \overline{\alpha}(\tau,s)d\tau}  ds \\
\qquad = 0 .
\end{array} \vspace{-5mm}
\end{equation} 
\end{remark}

\section{Horizontal contraction}
\label{sec:horizontal_contraction}

\subsection{Contraction and symmetries}
\label{sec:horizontal_inf_contraction}
Theorem \ref{thm:inf_contraction} 
guarantees contraction among the solutions of a system \emph{in every possible direction}. 
This result can be easily extended to capture contraction 
with respect to \emph{specific directions} -- a relevant feature
for contraction analysis in presence of symmetries like, 
for example, in synchronization problems.

The generalization of Theorem \ref{thm:inf_contraction} is based on
the introduction of horizontal Finsler-Lyapunov functions
on a manifold $\mathcal{M}$, whose {associated metrics $d$ 
(through bounds similar to \eqref{eq:FinsLyap_bounds}) are} tailored 
to the particular problem of interest. These functions are  positive
only on a suitably selected (horizontal) subspace $\mathcal{H}_x\subseteq T_x \mathcal{M}$, 
for each $x\in\mathcal{M}$, which characterize 
the set of directions (tangent vectors) taken into account by the Finsler structure.

\begin{definition}[\emph{Horizontal Finsler-Lyapunov function}]
\label{def:horizontal_LyapFins_function}
 Consider a manifold $\mathcal{M}$ of dimension $d$.
 For each $x\in\mathcal{M}$, { suppose that 
 $T_x\mathcal{M}$ can be subdivided} into a \emph{vertical distribution}
$\mathcal{V}_x\subset T_x\mathcal{M}$ 
 \begin{equation}
  \label{eq:V_x}
  \mathcal{V}_x := \mathrm{Span}(\{v_1(x),\dots, v_r(x)\}) \qquad 0\leq r<d \ ,
 \end{equation}
 and a \emph{horizontal distribution} $\mathcal{H}_x\subseteq T_x\mathcal{M}$ 
 complementary to $\mathcal{V}_x$, i.e.
 $\mathcal{V}_x \oplus \mathcal{H}_x = T_x\mathcal{M}$, 
 \begin{equation}
 \label{eq:H_x}
 \mathcal{H}_x := \mathrm{Span}(\{h_1(x),\dots, h_q(x)\}) \qquad 0 < q \leq d-r \,
 \end{equation}
 where $v_i$, $i \in \{1,\dots,r\}$, 
 and $h_i$, $i\in\{1,\dots,q\}$, are $C^1$ vector fields.
 
 A function $V:{T\mathcal{M}} \to \real_{\geq 0}$ that maps every 
 $(x,\delta x) \in {T\mathcal{M}}$ to $V(x,\delta x)\in \real_{\geq 0}$
 is a \emph{candidate horizontal Finsler-Lyapunov function for} \eqref{eq:sys} \emph{on} $\mathcal{H}_x$ if
 there exist $c_1,c_2\in \real_{\geq 0}$, $p\in\real_{\geq 1}$, 
 and a function $F:T\mathcal{M} \to \real_{\geq 0}$ such that \eqref{eq:FinsLyap_bounds}
 holds. Moreover, $V$ and $F$ satisfy the following conditions. 
 Given a set of isolated points $\Omega\subset\mathcal{M}$,
 \begin{itemize}
  \item[(i$_a$)] $V$ and $F$ are $C^1$ function for each $x\in\mathcal{M}$ 
  and $\delta x\in \mathcal{H}_x\setminus\{0\}$;
  \item[(i$_b$)]   $V$ and $F$ satisfy 
  $V(x,\delta x) = V(x,\delta x_h)$ and $F(x,\delta x) = F(x,\delta x_h)$
  for each $(x,\delta x)\in T\mathcal{M}$ such that 
  $(x,\delta x) = (x,\delta x_h) + (x,\delta x_v)$, 
  $\delta x_h \in \mathcal{H}_x$, and $\delta x_v \in \mathcal{V}_x$. 
  \item[(ii)] $F(x,\delta x) >0$ for each $x\in\mathcal{M}$ 
  and $\delta x\in \mathcal{H}_x\setminus\{0\}$.
  \item[(iii)] $F(x,\lambda \delta x) = \lambda F(x,\delta x)$   
  for each $\lambda>0$, $x\in\mathcal{M}$, 
   and $\delta x\in \mathcal{H}_x$; 
  \item[(iv)] $F(x,\delta x_1+\delta x_2) < F(x,\delta x_1) + F(x,\delta x_2)$ 
  for each $x\in \mathcal{M}$ and 
  $\delta x_1,\delta x_2\in \mathcal{H}_x\setminus\{0\}$ 
  such that $\delta x_1\neq \lambda \delta x_2$ for any given $\lambda \in \real$.
 \end{itemize} \vspace{-4.5mm}
\end{definition}
The conditions of Definition \ref{def:horizontal_LyapFins_function} 
resemble the conditions of Definition \ref{def:LyapFins_function},
particularized to horizontal tangent vectors {
$\delta x\in \mathcal{H}_x$}. 
The metric {induced by $F$ \eqref{eq:induced_distance_main}}
is only a pseudo-distance on $\mathcal{M}$ since two states $x_a,x_b\in\mathcal{M}$ may satisfy
$d(x_a,x_b)=0$ despite $x_a\neq x_b$. In fact, every piecewise differentiable curve
$\gamma\!:\!I\!\to\!\mathcal{M}$ that satisfies $\dot{\gamma}(s) \in \mathcal{V}_{\gamma(s)}$, 
for almost every $s\in I$,
also satisfies that 
$\int_I V(\gamma(s),\dot{\gamma}(s)) ds = 
{ \int_I F(\gamma(s),\dot{\gamma}(s)) ds} = 0$.
{By (i$_b$)}, the pseudo-distance $d$ measures the 
``distance'' between two given points $x_a$ and $x_b$
by considering only the horizontal component of curves
$\gamma:I\to\mathcal{M}$ connecting $x_a$ and $x_b$, 
that is, the component $\dot{\gamma}_h(s)$ of 
$\dot{\gamma}(s) = \dot{\gamma}_h(s) + \dot{\gamma}_v(s)$ where 
$\dot{\gamma}_h(s) \in\mathcal{H}_{\gamma(s)}$ and 
$\dot{\gamma}_v(s) \in\mathcal{V}_{\gamma(s)}$,
for each $s\in I$.

We can now provide the reformulation of Theorem \ref{thm:inf_contraction}
for horizontal Finsler-Lyapunov functions.

\begin{theorem}
\label{thm:horizontal_contraction}
Consider the system \eqref{eq:sys} on a smooth manifold 
$\mathcal{M}$ with $f$ of class $C^2$, 
a vertical distribution $\mathcal{V}_x$ \eqref{eq:V_x},
and a horizontal distribution $\mathcal{H}_x$ \eqref{eq:H_x}.
Let $\mathcal{C}\subseteq \mathcal{M}$ 
be a connected and forward invariant set
and ${\alpha}$ a function in $\real_{\geq 0} \to \real_{\geq 0}$.

Given a candidate horizontal Finsler-Lyapunov function $V$ 
for \eqref{eq:sys} on $\mathcal{H}_x$,
suppose that \eqref{eq:diff_contraction} holds 
for each {$t\in \real$}, each $x\in\mathcal{C}$
and each {$\delta x \in T_x\mathcal{M}$}. 
Then, the solutions to \eqref{eq:sys}
\begin{itemize}
\item[(i)]\emph{do not expand} the pseudo-distance $d$ \eqref{eq:induced_distance_main}
 on $\mathcal{C}$ 
if $\alpha(s) = 0$ for each $s\geq 0$: there exists $\gamma(s)\geq s$ such that
 $d(\psi_{t_0}(t,x_1),\psi_{t_0}(t,x_2)) \leq \gamma(d(x_1,x_2))$, 
 $\forall t_0\in \real$,  $\forall t>t_0$,$\forall x_1,x_2\in\mathcal{C}$;
 \item[(ii)]\emph{asymptotically contract} the pseudo distance $d$ on $\mathcal{C}$ 
 if $\alpha$ is a $\mathcal{K}$ function: (i) holds and 
 $\lim\nolimits\limits_{t\to \infty} d(\psi_{t_0}(t,x_1),\psi_{t_0}(t,x_2)) = 0$,
 $\forall t_0\in\real$, $\forall x_1,x_2\in\mathcal{C}$; 
  \item[(iii)]\emph{exponential contract} the pseudo distance $d$ on $\mathcal{C}$ 
if $\alpha(s) = \lambda s >0$ for each $s\geq 0$:
 there exists $K\geq 1$ s.t.
 $d(\psi_{t_0}(t,x_1),\psi_{t_0}(t,x_2)) \leq K e^{-\lambda(t-t_0)} d(x_1,x_2)$,
 $\forall t_0\in \real$,  $\forall t>t_0$,$\forall x_1,x_2\in\mathcal{C}$.
\end{itemize} \vspace{-5mm}
\end{theorem}
{
The next result particularizes Theorem \ref{thm:horizontal_contraction}
to the case in which the selected horizontal distribution is invariant 
along the dynamics of \eqref{eq:sys}. In coordinates, 
condition \eqref{eq:horizontal_invariance} below guarantees that 
$ \dot{\delta x_h} = \frac{\partial f(t,x)}{\partial x} \delta x_h$
along the solutions to \eqref{eq:sys}, which 
establishes the invariance of $\mathcal{H}_x$.
\begin{theorem}
 \label{thm:horizontal_contraction_plus_invariance}
Under the hypothesis of Theorem \ref{thm:horizontal_contraction},
consider the horizontal projection 
$\pi_x:{T_x\mathcal{M}} \to T_x\mathcal{M}$ 
that maps each $\delta x \in T_x \mathcal{M}$
to $\delta_h := \pi_x(\delta x) \in \mathcal{H}_x$.
Suppose that, in coordinates, 
$\forall t\in \real,  \forall (x,\delta x)\in T\mathcal{M},$
\begin{equation}
\label{eq:horizontal_invariance}
   \frac{\partial \pi_x(\delta x)}{\partial x}f(t,x) + 
   \frac{\partial \pi_x(\delta x)}{\partial \delta x} 
   \frac{\partial f(t,x)}{\partial x}\delta x 
   = \frac{\partial f(t,x)}{\partial x}\pi_x(\delta x);
\end{equation}
and suppose that \eqref{eq:diff_contraction} holds 
for each $t \in \real$, each $x\in\mathcal{C}$, 
and each $\delta x \in \mathcal{H}_x$. Then, 
the solutions to \eqref{eq:sys} satisfy
(i)-(iii) of Theorem \ref{thm:horizontal_contraction}.
\end{theorem}
\begin{proofof}\emph{Theorems \ref{thm:horizontal_contraction} and
\ref{thm:horizontal_contraction_plus_invariance}.}
The proof of Theorem \ref{thm:horizontal_contraction} is just
the repetition of the proof of Theorem \ref{thm:inf_contraction}
particularized to horizontal Finsler-Lyapunov functions.

The proof of Theorem \ref{thm:horizontal_contraction_plus_invariance}
exploits the identity \eqref{eq:horizontal_invariance} within the
argument of the proof of Theorem \ref{thm:inf_contraction}. 
For any given curve $\gamma:I\to\mathcal{C}$, let $\psi_{t_0}(\cdot,\gamma(s))$ 
be the solution to \eqref{eq:sys} from the initial condition $\gamma(s)$
at time $t_0$.
Using coordinates, 
define 
$x(t,s):=\psi_{t_0}(t,\gamma(s))$,
and $\delta x(t,s):= \frac{\partial}{\partial s}\psi_{t_0}(t,\gamma(s))$.
Consider the decomposition of 
$\delta x(t,s)$ into $\delta x(t,s) = \delta x_h(t,s)+ \delta x_v(t,s)$, 
respectively horizontal $\delta x_h(t,s) \in \mathcal{H}_{x(t,s)}$ 
and vertical $\delta x_v(t,s) \in \mathcal{V}_{x(t,s)}$ components. 
Note that $\delta x_h(t,s) = \pi_{x(t,s)}(\delta x(t,s))$. Therefore,
mimicking \eqref{eq:M04},
\begin{equation}
\label{eq:M04bis}
\begin{array}{rcl}
\! \frac{\partial}{\partial t}\delta x_h(t,s) 
 \!&\!=\!&\! \frac{\partial}{\partial t} \pi_{x(t,s)}(\delta x(t,s)) \\
  \!&\!=\!&\! \left[\frac{\partial\pi_x(\delta x)}{\partial x}_{|x(t,s),\delta x(t,s)}\right]f(t,x(t,s)) \ +\\
  \!&\!+\!&\! \left[\frac{\partial\pi_x(\delta x)}{\partial \delta x}_{|x(t,s),\delta x(t,s)}\right]
   \left[\frac{\partial f(t,x)}{\partial x}_{|x(t,s)}\right]\delta x(t,s) \\
  \!&\!=\!&\! \left[\frac{\partial f(t,x)}{\partial x}_{|x(t,s)}\right]
     \pi_{x(t,s)}(\delta x(t,s)) \\
 \!&\!=\!&\! \left[\frac{\partial f(t,x)}{\partial x}_{|x(t,s)}\right]\delta x_h(t,s),
\end{array}
\end{equation}
where the next to the last identity follows from \eqref{eq:horizontal_invariance}.

From the assumption (ii) in Definition \ref{def:horizontal_LyapFins_function},
$V(x(t,s),\delta x(t,s)) = V(x(t,s),\delta x_h(t,s))$, thus
$\frac{d}{dt}V(x(t,s),\delta x(t,s)) = \frac{d}{dt}V(x(t,s),\delta x_h(t,s)) $
for each $t\geq t_0$, and $s\in I$. 
Therefore, mimicking \eqref{eq:M05} and using \eqref{eq:M04bis}, 
and \eqref{eq:diff_contraction}, we get
\begin{equation}
\label{eq:M05bis}
\frac{d}{dt}V(x(t,s),\delta x_h(t,s)) 
\leq  -\alpha(V(x(t,s),\delta x_h(t,s)).
\end{equation}
From this inequality, the proof of Theorem \ref{thm:horizontal_contraction_plus_invariance}
continues as the proof of Theorem \ref{thm:inf_contraction} from \eqref{eq:M05}.
\end{proofof}
}

\begin{remark}
\label{rem:relax_horiz_c1}
The formulation of the LaSalle-like relaxations of Theorem \ref{thm:LaSalle_timeinvariant2} and
Remark \ref{rem:LaSalle} in Section \ref{sec:LaSalle} immediately extends
to horizontal Finsler-Lyapunov functions. Following Remark \ref{rem:piecewise_diff},
the regularity assumption (i) in Definition \ref{def:horizontal_LyapFins_function} 
can be relaxed to functions $V$ that are piecewise continuously differentiable and locally Lipschitz.
In such a case, the goal is to show that the inequality \eqref{eq:MAStab} holds.
This is guaranteed, for example, if the inequality in \eqref{eq:M05bis}
 holds for {almost every} $t$ and $s$.
\end{remark}

\subsection{Contraction on quotient manifolds}
\label{sec:quotient}

The notion of horizontal space is classical in the theory of quotient manifolds.
Let $\mathcal{M}$ be a given manifold and 
let $\mathcal{M}\setminus\!\!\sim$ be the \emph{quotient manifold} of $\mathcal{M}$ 
induced by the \emph{equivalence relation} $\sim \in \mathcal{M}\times\mathcal{M}$.
Given $x\in\mathcal{M}$, we denote by $[x]\in {\mathcal{M}\setminus\!\sim}$
the class of equivalence to $x$. 
Suppose that the system 
$\dot{x}=f(t,x)$ in \eqref{eq:sys} is a representation on $\mathcal{M}$ of a 
system on ${\mathcal{M}\setminus\!\sim}$ in the following sense:
for every $t_0\geq 0$, every $x_0$, and every $z_0 \in [x_0]$, 
the solution $\varphi_{t_0}(\cdot,z_0)$ to \eqref{eq:sys}
satisfies $\varphi_{t_0}(t,z_0) \in [\varphi_{t_0}(t,x_0)]$ for each $t\geq t_0$.
In such a case we call $\dot{x}=f(t,x)$ a \emph{quotient system} 
on ${\mathcal{M}\setminus\!\sim}$.
The equivalence relation $\sim$ usually describes the symmetries 
on the system dynamics on $\mathcal{M}$, which implicitly characterize  
the quotient dynamics. 
Every solution $\varphi_{t_0}(\cdot,z_0)$ of \eqref{eq:sys} 
from $z_0\in[x_0]\in{\mathcal{M}\setminus\!\sim}$
is a (lifted) representation of a unique solution $[\varphi_{t_0}(\cdot,x_0)]$
on the quotient manifold. 

The vertical space $\mathcal{V}_x$ at $x$ is defined as the tangent space to the fiber through $x$.
In this way, any tangent vector $\delta[x]$ to $T_{[x]} M\setminus\!\sim$ has a 
unique representation in the horizontal space $\mathcal{H}_x$, called the horizontal lift \cite{AbsMahSep2008}.
The particular selection of the vertical distribution guarantees that 
the horizontal Finsler-Lyapunov function $V$ on $\mathcal{H}_x$ is 
zero for each $\delta x\in\mathcal{V}_x$.
As a consequence $V$ and the induced pseudo-distance $d$
can be used to characterize the incremental properties of the
quotient system:
if the pseudo-distance $d$ on $\calM$ satisfies
\begin{equation}
\label{eq:quotient_almost_distance}
d(x_1,x_2) \neq 0 \qquad \forall x_1, x_2 \in \calM \mbox{ s.t. }
[x_1] \neq [x_2],
\end{equation}
then $d$ is a distance on ${\calM\setminus\!\sim}$ and 
asymptotic contraction of 
\eqref{eq:sys} on $\mathcal{M}$ is equivalent to 
incremental asymptotic stability of the quotient system on ${\mathcal{M}\setminus\!\sim}$,
implicitly represented by \eqref{eq:sys} on $\mathcal{M}$.
In fact, \eqref{eq:quotient_almost_distance} guarantees that 
$d:\calM\times\calM\to\real_{\geq 0}$ 
is a distance on ${{\calM\setminus\!\sim}}$
since $d([x_1],[x_2]):= \inf_{z_1\in[x_1],z_2\in[x_2]}d(z_1,z_2) = d(x_1,x_2)\neq 0$,
for each $x_1,x_2 \in \calM$ such that $[x_1]\neq[x_2]$.

Suppose that Theorem \ref{thm:horizontal_contraction} holds for 
a given quotient system \eqref{eq:sys}, and suppose that the induced
pseudo-distance satisfies \eqref{eq:quotient_almost_distance}.
Then, by considering the lifted solutions of \eqref{eq:sys} to ${\mathcal{M}\setminus\!\sim}$,
the system  \eqref{eq:sys} is
\emph{(i)}~\emph{incrementally stable} on $\mathcal{C}$ 
if $\alpha(s) = 0$ for each $s\geq 0$;
\emph{(ii)}~\emph{incrementally asymptotically stable} on $\mathcal{C}$ 
if $\alpha$ is a $\mathcal{K}$ function; and 
\emph{(iii)}~\emph{incrementally  exponential stable} on $\mathcal{C}$ 
if $\alpha(s) = \lambda s >0$ for each $s\geq 0$.
In this sense, horizontal contraction in the total space is a convenient way to study contraction on quotient systems.
\begin{remark}
A sufficient condition to guarantee that the pseudo-distance 
$d$ on $\calM$ is a distance on ${\calM\setminus\!\sim}$
is to require that $F$ in Definition
\ref{def:horizontal_LyapFins_function}
is a Finsler structure on ${\calM\setminus\!\sim}$.
For instance, remember that $\mathcal{V}_x$ at $x$ 
is defined as the tangent space to the fiber through $x$,
and call \emph{fiber function} any function
$g:\calM\to\calM$ that maps every $z\in [x]$ into $g(z) \in [x]$,
for each ${[x]\in\calM\setminus\!\sim}$. Then,
$F$ is a Finsler structure on ${\calM\setminus\!\sim}$
if $F(x,\delta x) = F(g(x),Dg(x)[\delta x])$
for any fiber function $g$ and any $(x,\delta x) \in T\calM$
(which establishes the invariance of $F$ along the fiber of the quotient manifold).
\end{remark}

Quotient systems are encountered in many applications including tracking, coordination, and synchronization.
The potential of horizontal contraction in such applications is illustrated by two popular examples.

\begin{example}[\emph{Consensus}]\\
\label{ex:consensus}We consider consensus algorithms of the form
\begin{equation}
\label{eq:cons_sys}
 \dot{x} = A(t)x  
\end{equation}
where $x\in\realn$ and, for each $t\geq 0$, 
$A(t)$ has nonnegative off-diagonal elements 
and row sums zero (we assume that $A(t)$ is continuously differentiable). 
These Metzler matrices \cite{Moreau2004} 
are typically used to model the graph topology of
network problems. Indeed, the $\delta$\emph{-graph of}$A(t)$ 
has an edge from the node $i$ to the node $j$, $i\neq j$, 
if $a_{ij}(t) \geq \delta \geq 0$.

Given $\mathbf{1} := \smallmat{1 & \dots & 1}^T$, the row sums equal to zero 
guarantee that $A(t)\mathbf{1} = 0$ for each $t\geq 0$. Indeed,
$\alpha \mathbf{1}$ is a consensus state of the network for every $\alpha \in \real$.
Because of this symmetry, \eqref{eq:cons_sys} represents a 
quotient system on the quotient manifold $\realn\setminus\sim$ constructed from 
the equivalence  $x\sim y$ iff  
$x-y = \alpha \mathbf{1}$, for some $\alpha\geq 0$.
In fact, if $x\sim y$ then $A(t)x= A(t)y$ for
each $t\geq 0$.
The elements of $\realn\setminus\sim$ are 
$[x] := \{x+\alpha \mathbf{1}\,|\, \alpha\in\real\}$,
the vertical space is given by 
$\mathcal{V}_x := \mathrm{Span}(\{\mathbf{1}\})$,
and the horizontal space can be taken as
$\mathcal{H}_x := 
\{\delta x\in\realn \,|\, \mathbf{1}^T \delta x=0\} 
= \mathcal{V}_x^{\perp}$.
{
 \eqref{eq:cons_sys} is also a time-varying monotone system \cite{Smith1995,Angeli2003},
 and its stability properties have been studied by many authors \cite{Moreau2004,Tsitsiklis86}.
Under uniform connectivity assumptions 
its solutions converge exponentially to the submanifold of equilibria given 
by  $[0]=\{\alpha \mathbf{1}\,|\, \alpha\in\real\}$, \cite[Section 2.2 and Theorem 1]{Moreau2004}. 
We revisit this classical example through a differential approach.}

Consider the displacements dynamics from \eqref{eq:cons_sys}  given by
$\dot{\delta x} = A(t) \delta x$, and the horizontal Finsler-Lyapunov function 
\begin{equation}
\label{eq:cons_V}
 V(x,\delta x) := \max\nolimits\limits_i\,\delta x_i - \min\nolimits\limits_i\,\delta x_i,
\end{equation}
that coincides with the classical consensus function adopted in \cite{Moreau2004,Tsitsiklis86} 
lifted to the tangent space.
See \cite{Sepulchre10b} for its relationship to the Hilbert projection metric, known to contract
along monotone mapping \cite{Bushell73}.
Note that $V$ satisfies every condition of 
Definition \ref{def:horizontal_LyapFins_function}
but continuous differentiability.
In particular, $V$ is positive and homogeneous
for every $\delta x\in \mathcal{H}_x$.
For $\delta x \in T_x\realn$,
$V(x,\delta x) = V(x,\delta x_h)$ with 
$\delta x_h $ horizontal component of $\delta x$,
since $V(x,\delta x_h+\alpha\mathbf{1}) 
= V(x,\delta x_h)$ for each $\alpha\in\real$. 

Following Remark \ref{rem:relax_horiz_c1}, the lack of differentiability is not an issue.
In fact, from \cite[Section 3.3]{Moreau2004}, for any initial condition 
$x_0\in\realn$ and any initial tangent vector $\delta x_0\in T_x\realn = \realn$, 
$V$ is non-increasing along the solution $\varphi_{t_0}(\cdot,x_0)$ to \eqref{eq:cons_sys},
namely $V(\varphi_{t_0}(t,x),D\varphi_{t_0}(t,x_0)[0,\delta x_0]) \leq V(x_0,\delta x_0)$ for each $t\geq t_0$.
This inequality is the result of the combination of \cite[Section 3.3]{Moreau2004}, showing that
$\max\nolimits\limits_i z_i - \min\nolimits\limits_i z_i$ is non-increasing for $\dot{z} = A(t)z$, 
and of the fact that the evolution $D\varphi_{t_0}(t,x_0)[0,\delta x_0]$ 
of $\delta x_0$ along the solution $\varphi_{t_0}(\cdot,x_0)$ 
is also a solution to the differential equation $\dot{\delta x} = A(t) \delta x$ (as shown in \eqref{eq:M04}).

By the same argument, exponential decreasing of $V$ is achieved under additional conditions on uniform connectivity 
on the adjacency matrix $A(t)$. Following \cite[Theorem 1]{Moreau2004}, 
define $A^*(t) := \int_t^{t+T} A(\tau) d\tau$ and suppose that 
there exist $k \in\{1,\dots,n\}$, $\delta >0$, and $T>0$ such that, 
for every $t\geq t_0$ and every $j \in\{1,\dots,n\}\setminus\{k\}$, there is a path
from the node $k$ to the node $j$ of the $\delta$-graph of $A^*(t)$. 
Then $V$ decreases exponentially along the solutions to \eqref{eq:cons_sys}. By integration, 
the quotient system defined by \eqref{eq:cons_sys} is incrementally exponentially
stable. As a corollary, every solution to the quotient system converges to the steady-state solution 
$[0]$, that is, every solution to \eqref{eq:cons_sys} exponentially converges to consensus.

The reader will notice that 
the incremental exponential stability of \eqref{eq:cons_sys} is
a straightforward consequence of the exponential stability results of \cite{Moreau2004}, through
the lifting to the tangent space of the (non-quadratic) Lyapunov function used in \cite{Moreau2004}. 
In this sense, the differential framework captures the equivalence 
on linear systems between stability and incremental stability. 
\end{example}

\begin{example}[\emph{Phase Synchronization}] \\
Consider the interconnection of 
$n$ agents $\dot{\theta}_k = u_k$, $\theta_k \in \mathbb{S}^1$ (phase),
given by
\begin{equation}
\label{eq:KUR1}
 \dot{\theta}_k = \frac{1}{n}\sum_{j=1}^n  \sin(\theta_j - \theta_k).
\end{equation}
Using $ s_{jk} := \sin(\theta_j - \theta_k)$, 
$ c_{jk} := \cos(\theta_j - \theta_k)$,
$\mathbf{1} := \smallmat{1 & \dots & 1}^T$,
the aggregate state $\theta := \smallmat{\theta_1 & \dots & \theta_n}^T$,
and the displacement vector 
$\delta\theta := \smallmat{\delta\theta_1 & \dots & \delta\theta_n}^T$,
\eqref{eq:KUR1} and the related displacement dynamics can be written as follows.
\begin{equation}
\label{eq:KUR2}
\begin{array}{l}
 \dot{\theta} = \underbrace{\frac{1}{n}\smallmat{
0 & s_{21} & \cdots & s_{n1} \\
s_{12} & 0 & \cdots & s_{n2} \\
\vdots & \vdots & \ddots & \vdots \\
s_{n1} & s_{n2} & \cdots & 0 \\
}}_{=:\mathbf{S}(\theta)}\mathbf{1}  
\\
\dot{\delta \theta} = \underbrace{\frac{1}{n}\smallmat{
-\sum\nolimits\limits_{j\neq 1} c_{j1} & c_{21} & \cdots & c_{n1} \\
c_{12} & -\sum\nolimits\limits_{j\neq 2} c_{j2}& \cdots & c_{n2} \\
\vdots & \vdots & \ddots & \vdots \\
c_{n1} & c_{n2} & \cdots & -\sum\nolimits\limits_{j\neq n} c_{jn}\\
}}_{=:\mathbf{C}(\theta)} \delta \theta.
\end{array}
\end{equation}

\eqref{eq:KUR2} is a quotient system based on the equivalence 
$\theta \sim \overline{\theta}$ iff there exists $\alpha \in \real$ such that $\theta-\overline{\theta}=\mathbf{1}\alpha$. 
In fact, $\mathbf{S}(\theta) = \mathbf{S}(\theta+\alpha\mathbf{1})$, which fixes the class of equivalence
$[\theta] = \{\theta+\alpha\mathbf{1} \,|\, \alpha \in \real \}$, and the vertical space $\mathcal{V}_\theta := \mathrm{Span}\{\mathbf{1}\}$.
As in the previous example we consider 
$\mathcal{H}_\theta :=  \mathcal{V}_\theta^{\perp} =
\{\delta \theta\in\realn \,|\,  \mathbf{1}^T \delta \theta=0\}$. 

{
Paralleling Example \ref{ex:oscillator}, we contrast the conclusions obtained with constant
and non-constant Finsler-Lyapunov functions.
It is well known that the open set $\mathcal{O}\subset\mathbb{S}^n$ given by phase vectors $\theta$ such that
$|\theta_j-\theta_k| < \frac{\Pi}{2}$ for each $j,k\in\mathbf\{1,\dots,n\}$,
is forward invariant. Thus,
 \eqref{eq:KUR2} contracts the horizontal
{constant} quadratic function $V(\theta,\delta \theta) := \delta \theta^T [I_n - \frac{\mathbf{1}\mathbf{1}^T}{n}] \delta \theta$ in $\mathcal{O}$,
as shown in \cite[Proposition 1]{Moreau2004}
($\mathbf{C}(\psi(t,\theta_0))$ is a symmetric Metzler matrix 
along solutions $\psi(t,\theta_0)$ for $\theta_0\in \mathcal{O}$).}
Almost global contraction can be established by considering 
the horizontal non-constant function given by the \emph{{non-constant}} metric
\begin{equation}
\label{eq:KUR3}
 V(\theta,\delta \theta) := \frac{1}{\rho^{2q}}\delta \theta^T \Pi \delta \theta,\qquad q\in \mathbb{N},
\end{equation} 
where $\Pi:= \left[I_n - \frac{\mathbf{1}\mathbf{1}^T}{n}\right]$
(note that $\Pi\delta\theta=0$ for $\delta \theta\in\mathcal{V}_\theta$), 
and $\rho$ is the magnitude of the \emph{centroid} 
$\rho e^{i\phi} := \frac{1}{n}\sum\nolimits\limits_{k=1}^n e^{i\theta_k}$.  Following \cite{Sepulchre07kuramoto}, $\rho\in[0,1]$
is a measure of synchrony of the phase variables, since $\rho$ is $1$ when all phases coincide, while $\rho$ is $0$ when
the phases are balanced. $\rho$ is also nondecreasing, since
$\dot{\rho} = \frac{\rho}{n} \sum\nolimits\limits_{k=1}^n \sin(\theta_k-\phi)^2$. In particular, $\dot{\rho}=0$ 
for $\rho =0$ (balanced phases) or for
$\sum\nolimits\limits_{k=1}^n \sin(\theta_k-\phi)^2=0$, which occurs on isolated 
critical points given by $n-m$ phases synchronized at $\phi + 2j\pi$ and $m$ 
phases synchronized at $\phi+\pi+2j\pi$, for $j\in \mathbb{N}$ and $0\leq m\leq \frac{n}{2}$
Synchronization is achieved for $m=0$, the other critical points are saddle points
(for an extended analysis see \cite[Section III]{Sepulchre07kuramoto}).

Using $\dot{V}$ to denote the left-hand side of \eqref{eq:diff_contraction}, we get
\begin{equation}
\label{eq:KUR4}
\begin{array}{rcl}
 \!\dot{V} 
 \!\!&\!\!=\!\!&\!\! \frac{1}{\rho^{2q}} \delta \theta^T \!\! \left( \!
-\frac{2q}{n} \sum\nolimits\limits_{k=1}^n \sin(\theta_k-\phi)^2 \Pi 
\!+\! \Pi \mathbf{C}(\theta) \!+\! \mathbf{C}(\theta) \Pi \!\right) \delta \theta \\
\!\!&\!\!=\!\!&\!\! \frac{2}{\rho^{2q}} \delta \theta^T \left(
-\frac{q}{n} \sum\nolimits\limits_{k=1}^n \sin(\theta_k-\phi)^2 \Pi 
+  \mathbf{C}(\theta)  \right) \delta \theta.
\end{array}
\end{equation}
For each $\theta\in\mathbb{S}^n$, $\dot{V} =0$ for $\delta \theta \in \mathcal{V}_\theta$.
$\dot{V}$ is negative for $\theta\in\mathcal{O}$ and $\delta \theta \in \mathcal{H}_\theta$. 
For $\theta \in  \mathbb{S}^n\setminus\mathcal{O}$ and $\delta \theta \in \mathcal{H}_\theta$,
$q$ can be suitably chosen to balance the presence of positive eigenvalues in $\mathbf{C}(\theta)$.
In fact, given any compact and forward invariant set $\mathcal{C}\subset \mathbb{S}^n$
that does not contain any balanced phase ($\rho = 0$) or saddle point
($\sum\nolimits\limits_{k=1}^n \sin(\theta_k-\phi)^2=0$), 
there exists a sufficiently small
$\varepsilon >0$ such that 
$\sum\nolimits\limits_{k=1}^n \sin(\theta_k-\phi)^2>\varepsilon$ and $\rho >0$ for every $\theta\in\mathcal{C}$. 
Thus, contraction on $\mathcal{C}$ is established by picking $q \geq \frac{2}{\varepsilon}$.

The pseudo-distance {induced by $F=\sqrt{V}$} on $\mathbb{S}^n$ is a distance 
on the quotient manifold $\mathbb{S}^n\setminus \mathbb{S}$. 
Thus, the analysis
above establishes incremental asymptotic stability of the 
quotient system represented by \eqref{eq:KUR1}
in every forward invariant region $\mathcal{C}$ that does not contain
the balanced phase point and saddle points.
\end{example}

 \begin{remark}
By splitting the 
 tangent bundle into a contracting (horizontal) and
 a non-contracting (vertical) sub-bundles,
 horizontal contraction makes contact to the theory
 of Anosov flows \cite{Smale1967,Plante1972}
 (extended to Finsler manifolds).
 The references  \cite{Mierczynski1991} and \cite{Mierczynski1995}
 provide early results on horizontal contraction, where
 Finsler structures 
 are exploited to
 study the asymptotic properties of cooperative systems 
 with a first integral, 
 namely a function $H:\mathcal{M}\to \real$,
 constant along the system dynamics.
 It is obvious that no contraction can be expected in directions transversal
 to the level sets of $H$. Those directions are excluded from the contraction 
 analysis by picking a horizontal distribution tangent to the level set.
 Likewise, results on synchronization based on the combination 
 of contraction analysis and systems symmetries (via projective metrics) 
 are proposed in \cite{Pham07} and \cite{Russo2011a}. For example, 
 convergence to flow-invariant linear submanifolds 
 is a key property for the analysis of synchronization 
 problems \cite[Section 3]{Pham07}, which is established by 
 contraction analysis on a suitably projected dynamics 
 \cite[Sections 2.2 and 2.3]{Pham07}. 
 \end{remark}
}

\subsection{Forward contraction}
\label{sec:attractors}

The use of horizontal contraction is not restricted to quotient systems or
systems with first integrals. We briefly
discuss in this section the concept of \emph{forward} contraction of $\dot{x} = f(x)$,
that we define as horizontal contraction for the particular case
\begin{equation}
\label{eq:horizontal_flow_direction}
\mathcal{H}_x := \mathrm{Span}(\{f(x)\}),\quad \mbox{for each }x\in\mathcal{M}.
\end{equation}
By definition, forward contraction captures the property that for every solution
$\varphi(\cdot,x_0)$ to $\dot{x}=f(x)$, $x_0\in\mathcal{M}$,
and every $T \geq 0$, the points $\varphi(t+T,x_0)$ and $\varphi(t,x_0)$ converge to each other
as $t\to\infty$. 
This property has strong implications for the limit set of $\dot{x} = f(x)$, as illustrated
by the following proposition.
Restricting the analysis to time-invariant systems $\dot{x}=f(x)$ for simplicity,
we propose a novel result on attractor analysis by exploiting 
forward contraction. 
{ 
The result take advantage of the fact that the horizontal distribution
$\mathcal{H}_x$ in \eqref{eq:horizontal_flow_direction} 
is invariant along the dynamics of the system,
in the sense of \eqref{eq:horizontal_invariance}
\footnote{
{ 
Using coordinates, take the 
projection $\pi_x(\delta x) := \sigma(x,\delta x) f(x)$, 
where
$\sigma(x,\delta x) := \frac{f^T(x)\delta x}{f^T(x)f(x)}$.
To establish \eqref{eq:horizontal_invariance}, note that 
$\frac{\partial \sigma(x,\delta x)}{\partial x} f(x)
+ \frac{\partial \sigma(x,\delta x)}{\partial \delta x} 
\frac{\partial f(x)}{\partial x}\delta x = 0$.
Therefore, 
$\frac{\partial \pi_x(\delta x)}{\partial x}f(x) + 
\frac{\partial \pi_x(\delta x)}{\delta \partial x} 
\frac{\partial f(x)}{\partial x}\delta x 
= \sigma(x,\delta x) \frac{\partial f(x)}{\partial x} f(x)
= \frac{\partial f(x)}{\partial x} \pi_x(\delta x)$.
}
}.
}

\begin{proposition}[\emph{Bendixson's like criterion}]
\label{prop:Bendixson}
Consider the system $\dot{x}=f(x)$
on a smooth manifold $\mathcal{M}$ with $f$ of class $C^2$, and
a forward invariant set $\mathcal{C}\subseteq \mathcal{M}$.
Given a $\mathcal{K}$ function $\alpha$
and a candidate horizontal Finsler-Lyapunov function on 
$\mathcal{H}_x$ in \eqref{eq:horizontal_flow_direction},
suppose that {  Theorem \ref{thm:horizontal_contraction_plus_invariance} } 
holds for $\dot{x}=f(x)$.
Then, no solution of $\dot{x}=f(x)$ in $\mathcal{C}$ is a periodic orbit.
\end{proposition}
\begin{proof}
Suppose that from $x_0\in\mathcal{C}$, the solution $\varphi(\cdot,x_0)$
is a periodic orbit $\Gamma$. 
Then, from the definition of
$\mathcal{H}_x$ and the continuity of $V$, 
there exist $m>0$ such that  
$m \leq V(x,f(x))$ for each $x\in\Gamma$ ($\Gamma$ is a compact set). 
From \eqref{eq:diff_contraction}, the definition $\mathcal{H}_x$, and
the fact that $\alpha$ is a function of class $\mathcal{K}$, 
there exists a class $\mathcal{KL}$ function $\beta$ such that
$ m \leq 
\lim\nolimits\limits_{t\to\infty} 
V(\psi(t,x_0),f(\psi(t,x_0))) 
\leq 
\lim\nolimits\limits_{t\to\infty} 
\beta(V(\psi(0,x_0),f(\psi(0,x_0))), t) =0 $.
A contradiction.
\end{proof}

 Forward contraction makes contact to a vast body of theory, primarily
 motivated by the Jacobian conjecture \cite{Chamberland97}.
 Conditions to establish the absence of periodic orbits are proposed in 
 \cite{Smith86} (see e.g. Theorem 7) and \cite{Muldowney90}, 
 and are based on specific matrix measures. 
 The connection to 
 Theorem \ref{thm:inf_contraction} can be established along the lines of 
 Section \ref{sec:literature_comparison}. 
 These conditions are generalized in \cite{Li93}, which connects the absence of
 periodic orbits to the contraction of a suitably defined functional
 $S$ in the manifold tangent bundle, as shown in \cite[Sections 2 and 3]{Li93}.
 In a similar way, Proposition \ref{prop:Bendixson} 
 relates the absence of periodic orbits to the contraction of 
 a horizontal Finsler-Lyapunov function $V$ on $\mathcal{H}_x =\mathrm{Span}\{f(x)\}$.
 Results on periodic orbits based on Finsler structures can be
 found already in the early work of \cite{Mierczynski1987}.

Under the assumption of boundedness of the solutions to $\dot{x}=f(x)$,
the absence of periodic orbit induced by the contraction argument is exploited 
in the next proposition to guarantee that a given set $\mathcal{A}$ is
asymptotically attractive.

\begin{proposition}[\emph{Asymptotic attractor on} $\mathcal{C}$]
\label{prop:asymptotic_attractiveness}
Consider the system $\dot{x}=f(x)$ on a smooth 
manifold $\mathcal{M}$ with $f$ of class $C^2$,
a forward invariant set $\mathcal{C}\subseteq \mathcal{M}$,
and a forward invariant set (attractor) $\mathcal{A} \subseteq \mathcal{\mathcal{C}}$.
Given a  $\mathcal{K}$ function $\alpha$
and a candidate horizontal Finsler-Lyapunov function on 
$\mathcal{H}_x$ in \eqref{eq:horizontal_flow_direction}, 
suppose that Theorem \ref{thm:horizontal_contraction_plus_invariance}  holds
for $\dot{x}=f(x)$, with the relaxed condition that 
\eqref{eq:diff_contraction} holds
for each $x\in\mathcal{C}\setminus \mathcal{A}$, 
and each $\delta x\in \mathcal{H}_x$. 
If
\begin{itemize}
\item $\mathcal{A}$ contains every equilibrium point $0=f(x)$, $x\in \mathcal{C}$; 
\item for every initial time $t_0$ and every initial condition $x_0\in\mathcal{C}$,
there exists a bounded set $\mathcal{U}_{x_0}\subseteq \mathcal{M}$ such that 
$\psi(t,x_0) \in \mathcal{U}_{x_0}$ for each $t\geq 0$,
\end{itemize}
then 
for every initial condition $x_0\in\mathcal{C}$,
and every neighborhood $\mathcal{U}\supset\mathcal{A}$, 
there exists $T_{(x_0,\mathcal{U})} \geq 0$ such that 
$\psi(t,x_0) \in \mathcal{U}$ for each $t\geq T_{(x_0,\mathcal{U})}$.
\end{proposition}
\begin{proof}
Since $\psi(t,x_0)$ belongs to the bounded set
$\mathcal{U}_{x_0}$ for each $t\geq 0$, by \cite[Lemma 4.1]{Khalil3}  
it converges to its $\omega$-limit set, given by the compact and forward
invariant set $\omega^+(x_0):=
\{x\in\mathcal{M}\,|\, x = \lim\nolimits\limits_{n\to\infty} \psi(t_n,x_0) 
\mbox{ where } t_n\in \real_{\geq 0} \to \infty \mbox{ as }n\to \infty  \}$.
Note that if $\lim\nolimits\limits_{t\to\infty}\psi(t,x_0) = x^*\in\mathcal{C}$
then, by hypothesis, $x^*$ belongs to $\mathcal{A}\subset\mathcal{U}$.
Therefore $\omega^+(x_0)\setminus\mathcal{A}$ does not contains equilibria.
We prove by contradiction that $\omega^+(x_0) \subseteq \mathcal{A}$.

Suppose that $\omega^+(x_0)\cap \mathcal{A} = \emptyset $.
By compactness of $\omega^+(x_0)$, the definition of
$\mathcal{H}_x$, and the continuity of $V$, 
there exist $m>0$ such that  
$m \leq V(x,f(x))$ for each $x\in\omega^+(x_0)$. 
Consider the solution $\psi(\cdot,x)$ whose
initial condition 
$x\in\omega^+(x_0)$. 
From \eqref{eq:diff_contraction}, the definition $\mathcal{H}_x$, and
the fact that $\alpha$ is a function of class $\mathcal{K}$, 
there exists a class $\mathcal{KL}$ function $\beta$ such that
$ m \leq 
\lim\nolimits\limits_{t\to\infty} 
V(\psi(t,x),f(\psi(t,x))) 
\leq 
\lim\nolimits\limits_{t\to\infty} 
\beta(V(\psi(0,x),f(\psi(0,x))), t) =0 $.
A contradiction.

Suppose that $\omega^+(x_0)\cap \mathcal{A} \neq \emptyset $
and $\omega^+(x_0)\not\subseteq \mathcal{A}$. By the same argument used above,
there exists a sequence of $t_k \in \real_{\geq 0}$ such that $t_k \to \infty$ as
$k \to \infty$ such that  $V(\psi(t_k,x),f(\psi(t_k,x))) \geq m > 0$ but
$\lim\nolimits\limits_{k\to \infty} V(\psi(t_k,x),f(\psi(t_k,x))) \leq 
\lim\nolimits\limits_{k\to\infty} 
\beta(V(\psi(0,x),f(\psi(0,x))), t_k) =0 $.
A contradiction.
\end{proof}

\section{Conclusions}
\label{sec:conclusion}

The paper introduces a differential Lyapunov framework for the analysis of incremental stability,
a property of interest in a number applications of nonlinear systems theory. Our main result
extends the classical Lyapunov theorem from stability to incremental stability by lifting the Lyapunov
function in the tangent bundle. In addition to classical Lyapunov conditions, Finsler-Lyapunov functions
endow the state space with a Finsler differentiable structure. Through integration along curves,  
the construction of a Finsler-Lyapunov function, a local object, implicitly provides the 
construction of  a decreasing distance between solutions, a global object.

The study of global distances through local metrics is the essence of Finsler geometry, 
a generalization of  Riemannian geometry. Several examples and applications in the paper 
suggest that the Finsler differentiable structure is indeed the natural framework for 
contraction analysis, unifying in a natural way earlier contributions restricted either to a Riemannian
framework \cite{Lohmiller1998,Rouchon2003} or to matrix measures of contraction \cite{Russo10,Sontag10yamamoto}. 
In the same way, the formulation of the results on differentiable manifolds
rather than in Euclidean spaces is not for the mere sake of generality but 
motivated by the  fact that global incrementally stability questions arising 
in applications involve nonlinear spaces as a rule rather than as an exception.

A central motivation to bridge Lyapunov theory and contraction analysis is to provide contraction analysis
with the whole set of system-theoretic tools derived from Lyapunov theory. The present paper only illustrates 
this program with LaSalle's Invariance principle but we expect many further generalizations of Lyapunov theory 
to carry out in the proposed framework. This includes the use of asymptotic methods such as averaging 
theory or singular perturbation theory (see e.g. the result \cite{DelVecchio2012} ), and, most importantly, the use of 
contraction analysis for the study of open and interconnected systems.
The original motivation for the present paper was to develop a differential 
framework for incremental dissipativity \cite{Angeli09,Jouffroy03asimple,Stan2007}
- \emph{differential dissipativity} - 
which will be the topic of a separate paper
(see e.g. \cite{Forni2013,Forni2013a} for
preliminary results developed while the current paper
was under review).

Although a straightforward extension of contraction, the concept of  horizontal contraction introduced in this paper
illustrates the potential of contraction analysis in areas only partially explored to date.  Primarily, it provides the natural
differential geometric framework to study contraction in systems with symmetries, disregarding variations in the symmetry directions
where no contraction is expected. Problems such as synchronization, coordination, observer design, and tracking all involve
a notion of horizontal contraction rather than contraction. The notion of forward contraction, which corresponds to the particular
case of selecting the vector field to span the horizontal distribution, connects the proposed framework 
to an entirely distinct theory which seeks to characterize asymptotic behaviors by Bendixson type of criteria, 
excluding periodic orbits or forcing convergence to equilibrium sets \cite{Li93}. 

Overall, we anticipate a number of  interesting developments beyond the basic theory presented in this paper and we hope that
the proposed differential framework will facilitate further bridges between differential geometry and Lyapunov theory, a continuing
source of inspiration for nonlinear control. 
\vspace{1mm}

\emph{Acknowledgements.} We thank J. Mierczy{\'n}ski for pointing us to 
the early references \cite{Mierczynski1987,Mierczynski1991,Mierczynski1995}
during the revision process of the paper. We thank also the anonymous reviewers. Their
comments were important to improve the original version of the paper.
\vspace{-1mm}
 \bibliographystyle{plain}

\end{document}